\documentclass[12pt]{article}

\usepackage{amsmath, amsthm, amssymb}
\usepackage{mathenv, mathrsfs, graphicx, verbatim}
\usepackage[left=1in,right=1in,top=1in,bottom=1in]{geometry}

\theoremstyle{plain}
\newtheorem{thm}{Theorem}
\newtheorem{coro}[thm]{Corollary}
\newtheorem{lem}[thm]{Lemma}

\newtheorem*{inclaim}{Claim}
\newtheorem{prop}[thm]{Proposition}

\theoremstyle{definition}
\newtheorem{defn}[thm]{Definition}

\DeclareMathOperator*{\CSP}{CSP}
\DeclareMathOperator*{\core}{core}
\DeclareMathOperator*{\Hom}{HOM}
\DeclareMathOperator*{\coCSP}{co-CSP}
\DeclareMathOperator*{\CS}{CanS}
\DeclareMathOperator*{\CL}{CanL}
\DeclareMathOperator*{\ar}{ar}

\DeclareMathOperator*{\lev}{level}
\DeclareMathOperator*{\hei}{height}

\DeclareMathOperator*{\Ex}{Ex}
\DeclareMathOperator*{\Eq}{Eq}

\newcommand{\mb}[1]{\mathbf{#1}}
\newcommand{\ccACO}{\ensuremath{\mathsf{AC^0}}}
\newcommand{\ccL}{\ensuremath{\mathsf{L}}}
\newcommand{\ccPL}{\ensuremath{\mathsf{\oplus L}}}
\newcommand{\ccNL}{\ensuremath{\mathsf{NL}}}
\newcommand{\ccP}{\ensuremath{\mathsf{P}}}
\newcommand{\ccNP}{\ensuremath{\mathsf{NP}}}

\newcommand{\one}{\ensuremath{\mathfrak{first}}}
\newcommand{\n}{\ensuremath{\mathfrak{last}}}
\newcommand{\suc}{\ensuremath{\mathfrak{suc}}}

\newcommand{\cC}{\mathcal{C}}

\newcommand{\cF}{\mathcal{F}}

\newcommand{\cL}{\mathcal{L}}

\newcommand{\cN}{\mathcal{N}}
\newcommand{\cO}{\mathcal{O}}
\newcommand{\cP}{\mathcal{P}}

\newcommand{\cS}{\mathcal{S}}

\newcommand{\cU}{\mathcal{U}}
\newcommand{\cV}{\mathcal{V}}

\newcommand{\bA}{\mathbf{A}}
\newcommand{\bB}{\mathbf{B}}
\newcommand{\bC}{\mathbf{C}}
\newcommand{\bD}{\mathbf{D}}
\newcommand{\bE}{\mathbf{E}}

\newcommand{\bG}{\mathbf{G}}
\newcommand{\bH}{\mathbf{H}}

\newcommand{\bK}{\mathbf{K}}

\newcommand{\bM}{\mathbf{M}}

\newcommand{\bP}{\mathbf{P}}
\newcommand{\bQ}{\mathbf{Q}}
\newcommand{\bR}{\mathbf{R}}
\newcommand{\bS}{\mathbf{S}}
\newcommand{\bT}{\mathbf{T}}

\newcommand{\bW}{\mathbf{W}}
\newcommand{\bX}{\mathbf{X}}
\newcommand{\bY}{\mathbf{Y}}
\newcommand{\bZ}{\mathbf{Z}}

\newcommand{\bbA}{\mathbb{A}}

\newcommand{\sD}{\mathscr{D}}
\newcommand{\sF}{\mathscr{F}}
\newcommand{\sS}{\mathscr{S}}
\newcommand{\sT}{\mathscr{T}}
\newcommand{\sP}{\mathscr{P}}

\begin{document}

\title{On Constraint Satisfaction Problems below P\footnote{Research supported by NSERC, FQRNT, and ERC Starting Grant PARAMTIGHT (No. 280152).}}

\author{L\'{a}szl\'{o} Egri\footnote{Institute for Computer Science and Control, Hungarian Academy of Sciences
(MTA SZTAKI), Budapest, Hungary. \texttt{\{laszlo.egri@mail.mcgill.ca\}}}
}

\date{}

\maketitle

\begin{abstract}
Symmetric Datalog, a fragment of the logic programming language Datalog, is conjectured to capture all constraint satisfaction problems (CSP) in \ccL. Therefore developing tools that help us understand whether or not a CSP can be defined in symmetric Datalog is an important task. It is widely known that a CSP is definable in Datalog and linear Datalog if and only if that CSP has bounded treewidth and bounded pathwidth duality, respectively. In the case of symmetric Datalog, Bulatov, Krokhin and Larose ask for such a duality (2008). We provide two such dualities, and give applications. In particular, we give a short and simple new proof of the result of Dalmau and Larose that ``Maltsev + Datalog $\Rightarrow$ symmetric Datalog'' (2008).

In the second part of the paper, we provide some evidence for the conjecture of Dalmau (2002) that every CSP in \ccNL\ is definable in linear Datalog. Our results also show that a wide class of CSPs--CSPs which do not have bounded pathwidth duality (e.g., the \ccP-complete {\sc Horn-3Sat} problem)--cannot be defined by any polynomial size family of monotone read-once nondeterministic branching programs.
\end{abstract}
\bigskip


\section{Introduction}

Constraint satisfaction problems ($\CSP$) constitute a unifying framework to
study various computational problems arising naturally in various branches of computer science, including artificial intelligence, graph homomorphisms, and database theory. Loosely speaking, an instance of a $\CSP$ consists of a list of variables and a set of constraints, each specified by an ordered tuple of variables and a constraint relation over some specified domain. The goal is then to determine whether variables can be assigned domain values such that all constraints are simultaneously satisfied.

Recent efforts have been directed at classifying the complexity of the so-called \emph{nonuniform} $\CSP$. For a fixed finite set of finite relations $\Gamma$, $\CSP(\Gamma)$ denotes the nonuniform $\CSP$ corresponding to $\Gamma$. The difference between an instance of $\CSP(\Gamma)$ and an instance of the general $\CSP$ is that constraints in an instance of $\CSP(\Gamma)$ take the form $(x_{i_1}, \dots, x_{i_k}) \in R$ for some $R \in \Gamma$. Examples of nonuniform $\CSP$s include \textsc{$k$-Sat}, \textsc{Horn-3Sat}, \textsc{Graph H-Coloring}, and many others.

For a relational structure $\bB$, the homomorphism problem $\Hom(\bB)$ takes a structure $\bA$ as input, and the task is to determine if there is a homomorphism from $\bA$ to $\bB$. For instance, consider structures that contain a single symmetric binary relation, i.e., graphs. A homomorphism from a graph $\bG$ to a graph $\bH$ is a mapping from $V_{\bG}$ to $V_{\bH}$ such that any edge of $\bG$ is mapped to an edge of $\bH$. If $\bH$ is a graph with a single edge then $\Hom(\bH)$ is the set of graphs which are two-colorable. There is a well-known and straightforward correspondence between the $\CSP$ and the homomorphism problem. For this reason, from now on we work only with the homomorphism problem instead of the $\CSP$. Nevertheless, we call $\Hom(\bB)$ a $\CSP$ and we also write $\CSP(\bB)$ instead of $\Hom(\bB)$, \mbox{as it is often done in the literature.}

The $\CSP$ is of course $\ccNP$-complete, and therefore research has focused on identifying
``islands'' of tractable $\CSP$s. The well-known $\CSP$ dichotomy conjecture of Feder and Vardi \cite{Feder/Vardi:1999:Computational} states that every $\CSP$ is either tractable or $\ccNP$-complete, and progress towards this conjecture has been steady during the last fifteen years. From a complexity-theoretic perspective, the classification of $\CSP(\bB)$ as in $\ccP$ or being $\ccNP$-complete is rather coarse and therefore somewhat dissatisfactory. Consequently, understanding the fine-grained complexity of $\CSP$s gained considerable attention during the last few years. Ultimately, one would like to know the precise complexity of a $\CSP$ lying in $\ccP$, i.e., to identify a ``standard'' complexity class for which a given $\CSP$ is complete. Towards this, it was established that Schaefer's $\ccP-\ccNP$ dichotomy for Boolean $\CSP$s \cite{Schaefer:1978:Complexity} can indeed be refined: each $\CSP$ over the Boolean domain is either definable in first order logic, or complete for one of the classes $\ccL$, $\ccNL$, $\ccPL$, $\ccP$ or $\ccNP$ under \emph{\ccACO-reductions} \cite{Allender/et_al:09:Complexity}. The question whether some form of this fine-grained classification extends to non-Boolean domains is rather natural. The two most important tools to study $\CSP$s whose complexity is below $\ccP$ are \emph{symmetric Datalog} and \emph{linear Datalog}, syntactic restrictions of the database-inspired logic programming language \emph{Datalog}. We say that $\coCSP(\bB)$--the complement of $\CSP(\bB)$--is definable in (linear, symmetric) Datalog if the set of structures that do not homomorphically map to $\bB$ is accepted by a (linear, symmetric) Datalog program.\footnote{The reason we define $\coCSP(\bB)$ instead of $\CSP(\bB)$ in (linear, symmetric) Datalog is a technicality explained in Section~\ref{Defining_CSPs}.}

Symmetric Datalog programs can be evaluated in logarithmic space ($\ccL$), and in fact, it is conjectured that if $\coCSP(\bB)$ is in $\ccL$ then it can also be defined in symmetric Datalog \cite{Egri/Larose/Tesson:07:Symmetric}. There is a considerable amount of evidence supporting this conjecture (see, for example, \cite{Egri/Larose/Tesson:07:Symmetric,Egri/et_al:2011:Complexity,Dalmau/Larose:2008:Maltsev,Larose/Tesson:09:Universal,Carvalho/et_al:11:Maltsev}), and therefore providing tools to show whether $\coCSP(\bB)$ can be defined in symmetric Datalog is an important task. It is well known and easy to see that for any structure $\bB$, there is a set of structures $\cO$, called an \emph{obstruction set}, such that a structure $\bA$ \emph{homomorphically maps to} $\bB$ if and only if there is no structure in $\cO$ that homomorphically maps to $\bA$. In fact, there are many possible obstruction sets for any structure $\bB$. We say that $\bB$ has \emph{duality} $X$, if $\bB$ has an obstruction set which has the special property $X$. The following two well-known theorems relate definability of $\coCSP(\bB)$ in Datalog and linear Datalog to $\bB$ having \emph{bounded treewidth} and \emph{bounded pathwidth} duality, respectively:
\begin{enumerate}
\item $\coCSP(\bB)$ is definable in Datalog if and only if $\bB$ has bounded treewidth duality \cite{Feder/Vardi:1999:Computational};
\item $\coCSP(\bB)$ is definable in linear Datalog if and only if $\bB$ has bounded pathwidth duality \cite{Dalmau:02:Constraint}.
\end{enumerate}

It was stated as an open problem in \cite{Bulatov/Krokhin/Larose:08:Dualities}\ to find a duality for symmetric Datalog in the spirit of the previous two theorems. We provide two such dualities: \emph{symmetric bounded pathwidth duality} (SBPD) and \emph{piecewise symmetric bounded pathwidth duality} (PSBPD). We note that SBPD is a special case of PSBPD. For both bounded treewidth and bounded pathwidth duality, the structures in the obstruction sets are restricted to have some special form. For SBPD and PSBPD the situation is a bit more subtle. In addition that we require the obstruction sets to contain structures only of a special form (they must have bounded pathwidth), the obstruction sets must also possess a certain ``symmetric closure'' property. To the best of our knowledge, this is the first instance of a duality where in addition to the local requirement that each structure must be of a certain form, the set must also satisfy an interesting global requirement.

Using SBPD, we give a short and simple new proof of the main result of \cite{Dalmau/Larose:2008:Maltsev} that ``Maltsev + Datalog $\Rightarrow$ symmetric Datalog''. Considering the simplicity of this proof, we suspect that SBPD (or PSBPD) could be a useful tool in an attempt to prove the \emph{algebraic symmetric Datalog conjecture} \cite{Larose/Tesson:09:Universal}, a conjecture that proposes an algebraic characterization of all $\CSP$s lying in $\ccL$. An equivalent form of this conjecture is that ``Datalog + $n$-permutability $\Rightarrow$ symmetric Datalog'' (by combining results from \cite{Hobby/McKenzie:1988:Structure,Barto/Kozik:09:Constraint,Larose/Zadori:07:Bounded}), where $n$-permutability is a generalization of Maltsev.

One way to gain more insight into the dividing line between $\CSP$s in $\ccL$ and $\ccNL$ is through studying the complexity of $\CSP$s corresponding to oriented paths. It is known that all these $\CSP$s are in $\ccNL$ (by combining results from \cite{Feder:01:Classification,Dalmau/Krokhin:2008:Majority,Dalmau:02:Constraint}), and it is natural to ask whether there are oriented paths for which the $\CSP$ is $\ccNL$-complete and $\ccL$-complete. We provide two classes of oriented paths, $\cC_1$ and $\cC_2$, such that for any $\bB_1 \in \cC_1$, the corresponding $\CSP$ is $\ccNL$-complete, and for any $\bB_2 \in \cC_2$, the corresponding $\CSP$ is in $\ccL$. In fact, it can be seen with the help of \cite{Larose/Tesson:09:Universal} that for most $\bB_2 \in \cC_2$, $\CSP(\bB_2)$ is $\ccL$-complete. To prove the membership of $\CSP(\bB_2)$ in $\ccL$ (for $\bB_2 \in \cC_2$), we use PSBPD in an essential way. One can hope to build on this work to achieve an $\ccL$-$\ccNL$ dichotomy for oriented paths.

In the second part of the paper, we investigate $\CSP$s in $\ccNL$. Based on the observation that any $\CSP$ known to be in $\ccNL$ is also known to be definable by a linear Datalog program, Dalmau conjectured that every $\CSP$ in $\ccNL$ can be defined by a linear Datalog program \cite{Dalmau:02:Constraint}. Linear Datalog(\suc,$\neg$) (linDat(\suc,$\neg$)) denotes the extension of linear Datalog in which we allow negation and access to an order over the domain of the input. It is known that any problem in $\ccNL$ can be defined by a linDat(\suc,$\neg$) program \cite{Dalmau:02:Constraint,Gradel:1992:Capturing,Immerman:1999:Descriptive}, and therefore one way to prove the above conjecture would be to show that any $\CSP$ that can be defined by a linDat(\suc,$\neg$) program can also be defined by a linear Datalog program. We consider a restriction of the conjecture because proving it in its full generality would separate $\ccNL$ from $\ccP$ (using \cite{Afrati/Cosmadakis:89:Expressiveness}).

\emph{Read-once linear Datalog(\suc)} (1-linDat(\suc)) is a subclass of linDat(\suc,$\neg$), but a subclass that has interesting computational abilities, and for which we are able to find the chink in the armor. We can easily define some $\ccNL$-complete problems in 1-linDat(\suc), such as the $\CSP$ directed $st$-connectivity ($st$-\textsc{Conn}), and also problems that are not \emph{homomorphism-closed}, such as determining if the input graph is a clique on $2^n$ vertices, $n \geq 1$. Because any problem that can be defined with a linear Datalog program must be homomorphism closed, it follows that 1-linDat(\suc) can define nontrivial problems which are in $\ccNL$ but which are not definable by any linear Datalog program. However, our main result shows that if $\coCSP(\bB)$ can be defined by a 1-linDat(\suc) program, then $\coCSP(\bB)$ can also be defined by a linear Datalog program. The crux of our argument applies the general case of the Erd\H{o}s-Ko-Rado theorem to show that a 1-linDat(\suc) program does not have enough ``memory'' to handle structures of unbounded pathwidth.


Our proof establishing the above result for 1-linDat(\suc) programs can be adapted to show a parallel result for a subclass of \emph{nondeterministic branching programs}, which constitute an important and well-studied class of computational models (see the book \cite{Wegener:00:Branching}). More precisely, we show that if $\coCSP(\bB)$ can be defined by a \emph{poly-size family of read-once\footnote{Our read-once restriction for nondeterministic branching programs is less stringent than the usual restriction because we require the programs to be read-once only on certain inputs.} monotone nondeterministic} branching programs (mnBP1(poly)) then $\coCSP(\bB)$ can also be defined by a linear Datalog program.\footnote{A 1-linDat(\suc) can be converted into an mnBP1(poly), so another way to present our results would be to do the proofs in the context of mnBP1s, and then to conclude the parallel result for 1-linDat(\suc).}

Finally, our results can be interpreted as lower-bounds on a wide class of $\CSP$s: if $\bB$ does not have bounded pathwidth duality, then $\coCSP(\bB)$ cannot be defined with any 1-linDat(\suc) program or with any mnBP1(poly). A specific example of such a $\CSP$ would be the $\ccP$-complete \textsc{Horn-3Sat} problem, and more generally, Larose and Tesson showed that any $\CSP$ whose associated \emph{variety admits the unary, affine or semilattice types} does not have bounded pathwidth duality (see \cite{Larose/Tesson:09:Universal} for details).


\section{Preliminaries}


\subsection{Basic Definitions}

A \emph{vocabulary} (or \emph{signature}) is a finite set of relation symbols with associated arities. The arity function is denoted with $\ar(\cdot)$. If $\bA$ is a relational structure over a vocabulary $\tau$, then $R^\bA$ denotes the relation of $\bA$ associated with the symbol $R \in \tau$. The lightface equivalent of the name of the structure denotes the universe of the structure, e.g., the universe of $\bA$ is $A$.

A \emph{tuple structure} $\tilde{\bA}$ over a vocabulary $\tau$ is a pair $(\tilde{A}, T)$: $T$ is a set of pairs of the form $(R,\mb{t})$, where $R \in \tau$ and $\mathbf{t}$ is an $\ar(R)$-tuple, and $\tilde{A}$ is the domain of $\tilde{A}$, i.e., $\tilde{A}$ contains every element that appears in some tuple $\mb{t}$, and possibly some other elements. Slightly abusing notation, we write $(R,\mb{t}) \in \tilde{\bA}$ to mean $(R,\mb{t}) \in T$, where $\tilde{\bA} = (\tilde{A}, T)$. Clearly, tuple structures are equivalent to relational structures. If $\bA$ is a relational structure, we denote the equivalent tuple structure with $\tilde{\bA}$, and vice versa. For convenience, we use the two notations interchangeably. We note that \emph{all} structures in this paper are finite.

Let $\bB$ be a structure of the same signature $\tau$ as $\bA$. The \emph{union} $\bA \cup \bB$ of $\bA$ and $\bB$ is the $\tau$-structure whose universe is $A \cup B$, and for each $R \in \tau$, $R^{\bA \cup \bB}$ is defined as $R^{\bA} \cup R^{\bB}$. (Note that it is possible that $A \cap B \neq \emptyset$.) A {\em homomorphism} from $\bA$ to $\bB$ is a map $f$ from $A$ to $B$ such that $f(R^\bA) \subseteq R^\bB$ for each $R \in \tau$. If there exists a homomorphism from $\bA$ to $\bB$, we often denote it with $\bA \rightarrow \bB$. If that homomorphism is $f$, we write $\bA \overset{f}{\longrightarrow} \bB$. A structure is called a \emph{core}
if it has no homomorphism to any of its proper substructures.
A \emph{retract} of a structure $\bB$ is an induced substructure $\bB'$ of $\bB$ such that there is a homomorphism $g\;:\; B \rightarrow B'$ with $g(b) = b$ for every $b \in B'$. A retract of $\bB$ that has minimal size among all retracts of $\bB$ is called a core of $\bB$. It is well known that all cores of a structure are isomorphic, and so one speaks of the core of a structure $\bB$, $\core(\bB)$.

We denote by $\CSP(\bB)$ the set $\{\bA \;|\; \bA \text{ is a $\tau$-structure such that $\bA \rightarrow \bB$}\}$, and by $\coCSP(\bB)$ the complement of $\CSP(\bB)$, i.e., the set $\{\bA \;|\; \bA \text{ is a $\tau$-structure such that $\bA \not\rightarrow \bB$}\}$. If we are given a class of $\tau$-structures $\cC$ such that for any $\bA \in \cC$, and any $\bB$ such that $\bA \rightarrow \bB$ it holds that $\bB \in \cC$, then we say that $\cC$ is \emph{homomorphism-closed}. \emph{Isomorphism closure} is defined in a similar way.

An $n$-ary operation on a set $A$ is a map $f:A^n \rightarrow A$. Given an $h$-ary relation $R$ and an  $n$-ary operation $f$ on the same set $A$, we say that $f$ {\em preserves} $R$ or that $R$ is {\em invariant} under $f$ if the following holds: given any matrix $M$ of size $h \times n$ whose columns are in $R$, applying $f$ to the rows of $M$ produces an $h$-tuple in $R$. A {\em polymorphism} of a structure $\bB$ is an operation $f$ that preserves each relation in $\bB$.

\begin{defn}[Maltsev Operation]
A ternary operation $f : A^3 \rightarrow A$ on a finite set $A$ is called a Maltsev operation if it satisfies the following identities: $f(x,y,y) = f(y,y,x) = x, \forall x,y \in A.$
\end{defn}

\subsection{Datalog}\label{Datalog}

We provide only an informal introduction to Datalog and its fragments, and the reader can find more details, for example, in \cite{Libkin:2004:Elements,Dalmau:02:Constraint,Egri/Larose/Tesson:07:Symmetric}. Datalog is a database-inspired query language whose connection with $\CSP$-complexity is now relatively well understood (see, e.g., \cite{Barto/Kozik:09:Constraint}). Let $\tau$ be some finite vocabulary. A Datalog program over $\tau$ is specified by a finite set of rules of the form $h \leftarrow b_1 \wedge \dots \wedge b_t$, where $h$ and the $b_i$ are atomic formulas $R(x_1,\dots,x_k)$. When we specify the variables of an atomic formula, we always list the variables from left to right, or we simply provide a tuple $\mathbf{x}$ of variables whose $i$-th variable is $\mathbf{x}[i]$. We distinguish two types of relational predicates occurring in a Datalog program: predicates that occur at least once in the head of a rule (i.e., its left-hand side) are called \emph{intensional database predicates} (\emph{IDBs}) and are not in $\tau$. The predicates which occur only in the body of a rule (its right-hand side) are called \emph{extensional database predicates} (\emph{EDBs}) and must all lie in $\tau$. A rule that contains no IDB in the body is called a \emph{nonrecursive rule}, and a rule that contains at least one IDB in the body is called a \emph{recursive rule}. A Datalog program contains a distinguished IDB of arity $0$ which is called the \emph{goal predicate}; a rule whose head IDB is a goal IDB is called a \emph{goal rule}.

\emph{Linear Datalog} is a syntactic restriction of Datalog in which there is at most one IDB in the body of each rule. The class of linear Datalog programs that contains only rules with at most $k$ variables and IDBs with at most $j \leq k$ variables is denoted by \emph{linear $(j,k)$-Datalog}. We say that the \emph{width} of such a linear Datalog program is $(j,k)$.

\emph{Symmetric Datalog} is a syntactic restriction of linear Datalog. A linear Datalog program $\cP$ is symmetric if for any recursive rule $I(\mathbf{x}) \leftarrow J(\mathbf{y}) \wedge \bar{E}(\mathbf{z})$ of $\cP$ (except for goal rules), where $\bar{E}(\mathbf{z})$ is a shorthand for the conjunction of the EDBs of the rule over variables in $\mathbf{z}$, the symmetric pair $J(\mathbf{y}) \leftarrow I(\mathbf{x}) \wedge \bar{E}(\mathbf{z})$ of that rule is also in $\cP$. The \emph{width} of a symmetric Datalog program is defined similarly to the width of a linear Datalog program.

We explain the semantics of linear (symmetric) Datalog using \emph{derivations} (it could also be explained with \emph{fixed point operators}, but that would be inconvenient for the proofs). Let $\cP$ be a linear Datalog program with vocabulary $\tau$. A \emph{$\cP$-derivation with codomain $D$} is a sequence of pairs $\sD = (\rho_1, \lambda_1), \dots, (\rho_q,\lambda_q)$, where $\rho_\ell$ is a rule of $\cP$, and $\lambda_\ell$ is a function from the variables $V_\ell$ of $\rho_\ell$ to $D$, $\forall \ell \in [q]$. The sequence $\mathscr{D}$ must satisfy the following properties. Rule $\rho_1$ is nonrecursive, and $\rho_q$ is a goal rule. For all $\ell \in [q-1]$, the head IDB $I$ of $\rho_\ell$ is the IDB in the body of $\rho_{\ell+1}$,
and if the variables of $I$ in the head of $\rho_\ell$ and the body of $\rho_{\ell+1}$ are $\mathbf{x}$ and $\mathbf{y}$, respectively, then $\lambda_\ell(\mathbf{x}[i]) = \lambda_{\ell+1}(\mathbf{y}[i])$, $\forall i \in [\ar(I)]$.

Let $\sD$ be a derivation. Let $R(\mathbf{z})$ be an EDB (with variables $\mathbf{z}$) appearing in some rule $\rho_\ell$ of $\sD$. We write $R(\mathbf{t})$ to denote that $\lambda_\ell(\mathbf{z}) = \mathbf{t}$, i.e., that $\lambda_\ell$ \emph{instantiates} the variables of $R(\mathbf{z})$ to $\mathbf{t}$. If $R(\mathbf{z})$ appears in some rule $\rho_\ell$ of $\sD$ and $\lambda_\ell(\mathbf{z}) = \mathbf{t}$, we say that $R(\mathbf{t})$ \emph{appears} in $\rho_\ell$, or less specifically, that $R(\mathbf{t})$ \emph{appears} in $\sD$.

Given a structure $\bA$ and a derivation $\sD$ with codomain $A$ for a program $\cP$, we say that \emph{$\sD$ is a derivation for $\bA$} if for every $R(\mathbf{t})$ that appears in a rule of $\sD$, $(R,\mathbf{t}) \in \tilde{\bA}$. The notation for a $\cP$-derivation for a structure $\bA$ will have the form $\sD_{\cP}(\bA)$. A linear (symmetric) Datalog program $\cP$ \emph{accepts} an input structure $\bA$ if there exists a $\cP$-derivation for $\bA$.

 \begin{defn}[Read-once Derivation]
 We say that a derivation $\sD$ is \emph{read-once} if every $R(\mathbf{t})$ that appears in $\sD$ appears exactly once in $\sD$, except when $R$ is the special EDB \suc, \one, or \n, defined in Section~\ref{NL_section}.
 \end{defn}

An example is given in Figure~\ref{example}. The vocabulary is $\tau = \{E^2,S^1,T^1\}$, where the superscripts denote the arity of the symbols. Notice that in the symmetric Datalog program $\cP$, rules (2) and (3) form a symmetric pair. It is not difficult to see that $\cP$ accepts a $\tau$-structure $\bA$ if and only if there is an oriented path (see Section~\ref{L_defs}) in $E^\bA$ from an element in $S^\bA$ to an element in $T^\bA$.

\begin{figure}[h!]
\noindent\begin{minipage}[h!]{0.32\textwidth}
\begin{eqnarray}[rclc]
I(x) &\leftarrow& S(x)&\qquad \qquad\\
I(y) &\leftarrow& I(x)\wedge E(x,y)\\
I(x) &\leftarrow& I(y)\wedge E(x,y)\\
G &\leftarrow& I(x)\wedge T(x)
\end{eqnarray}
\end{minipage}\begin{minipage}[h!]{0.06\textwidth}
\hfill
\end{minipage}\begin{minipage}[h!]{0.62\textwidth}
\begin{flushright}\includegraphics[scale=1]{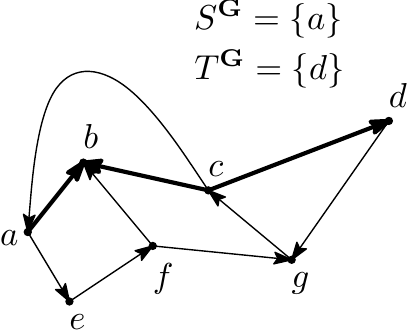}\end{flushright}
\end{minipage}
\bigskip
\begin{center}
\includegraphics[scale=1]{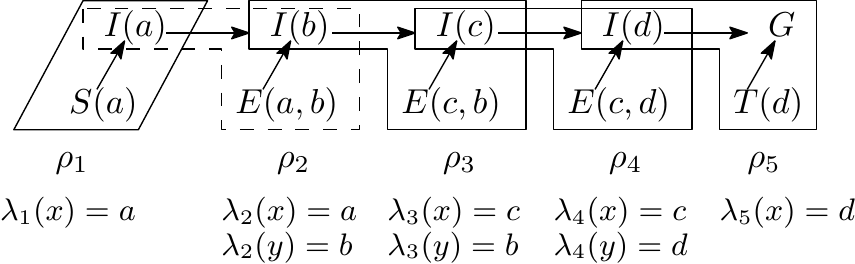}
 \caption{\emph{Top left:} Symmetric Datalog program $\cP$. \emph{Top right:} Input structure $\bG$ where the binary relation $E^\bG$ is specified by the digraph. \emph{Bottom:} Visualization of a $\cP$-derivation $\sD_{\cP}(\bG)=(\rho_1,\lambda_1),\dots,(\rho_5,\lambda_5)$ for $\bG$, where $\rho_1$ is the nonrecursive rule (1), $\rho_2,\rho_4$ are rule (2), $\rho_3$ is rule (3), and $\rho_5$ is the goal rule (4). In the diagram, for example, the dashed box corresponds to rule $\rho_2$, and it is the rule $I(y) \leftarrow I(x)\wedge E(x,y)$ of $\cP$, where $\lambda_2$ assigns $a$ to variable $x$ and $b$ to variable $y$. Observe that $\sD_{\cP}(\bG)$ is read-once.}\label{example}
\end{center}
\end{figure}

\subsection{Path Decompositions and Derivations}

\begin{defn}\label{path_decomp}[Path Decomposition]
Let $\bS$ be a $\tau$-structure. A \emph{$(j,k)$-path decomposition} (or path decomposition of \emph{width $(j,k)$}) of $\bS$ is a sequence $S_0,\dots,S_{n-1}$ of subsets of $S$ such that
\begin{enumerate}
\item For every $(R,(a_1,\dots,a_{\ar(R)})) \in \tilde{\bS}$, $\exists \ell \in \{0,\dots,n-1\}$ such that $\{a_1,\dots,a_{\ar(R)}\} \subseteq S_\ell$;
\item If $a \in S_i \cap S_{i'}$ ($i < i'$) then $a \in S_\ell$ for all $i < \ell < i'$;\label{inherit}
\item $\forall \ell \in \{0,\dots,n-1\}$, $|S_\ell| \leq k$, and $\forall \ell \in \{0,\dots,n-2\}$, $|S_\ell \cap S_{\ell+1}| \leq j$.\label{jk}
\end{enumerate}
\end{defn}

For ease of notation, it will be useful to introduce a concept closely related to path decompositions. Let $\tau$ be a vocabulary. Let $\bS$ be a $\tau$-structure that can be expressed as \mbox{$\bS = \bS_0 \cup \dots \cup \bS_{n-1}$}, where the $S_0,\dots,S_{n-1}$ (the universes of the $\bS_i$) satisfy properties \ref{inherit} and \ref{jk} above. We say that $\bS$ is a \emph{$(j,k)$-path}, and that $(\bS_0, \dots, \bS_{n-1})$ is a \emph{$(j,k)$-path representation} of $\bS$. We denote $(j,k)$-path representations with script letters, e.g., $\mathscr{S} = (\bS_0, \dots, \bS_{n-1})$. The substructure $\bS_i \cup \dots \cup \bS_{i'}$ of $\bS$ (assuming a $(j,k)$-representation is fixed) is denoted by $\bS_{[i,i']}$. We call $n$ the \emph{length} of the representation. Obviously, a structure is a $(j,k)$-path if and only if it admits a $(j,k)$-path decomposition.

Let $\sD = (\rho_1, \lambda_1), \dots, (\rho_q,\lambda_q)$ be a derivation for some linear or symmetric program $\cP$ with vocabulary $\tau$. We can extract from $\sD$ a $\tau$-structure $\Ex(\sD)$ such that $\sD$ is a derivation for $\Ex(\sD)$. We specify $\Ex(\sD)$ as a tuple structure $\tilde{\bA}$: for each $R(\mathbf{t})$ that appears in $\sD$ ($R \in \tau$), we add the pair $(R,\mathbf{t})$ to $\tilde{\bA}$, and set $\tilde{A}$ to be the set of those elements that appear in a tuple.

Let $\sD = (\rho_1, \lambda_1), \dots, (\rho_q,\lambda_q)$ be a derivation. For each $x$ that is in a rule $\rho_\ell$ for some $\ell \in [q]$, call $x^\ell$ the \emph{indexed version of $x$}. We define an equivalence relation $\Eq(\sD)$ on the set of indexed variables of $\sD$. First we define a graph $G = (V,E)$ as:
\begin{itemize}
\item $V$ is the set of all indexed versions of variables in $\sD$;
\item $(x^\ell, y^{\ell'}) \in E$ if $\ell' = \ell + 1$, $x$ is the $i$-th variable of the head IDB $I$ of $\rho_\ell$, and $y$ is the $i$-th variable of the body IDB $I$ of $\rho_{\ell+1}$.
\end{itemize}
Two indexed variables $x^\ell$ and $y^{\ell'}$ are related in $\Eq(\sD)$ if they are connected in $G$. Observe that if $C = \{x_1^{\ell_1}, x_2^{\ell_2}, \dots, x_c^{\ell_c}\}$ is a connected component of $G$, then it must be that $\lambda_{\ell_1}(x_1) = \lambda_{\ell_2}(x_2) = \dots = \lambda_{\ell_c}(x_c)$.
\begin{defn}[Free Derivation]\label{free_der}
Let $\cP$ be a linear Datalog program and $\sD = (\rho_0,\lambda_0), \dots,$ $(\rho_q,\lambda_q)$ be a derivation for $\cP$. Then $\sD$ is said to be \emph{free} if for any two $(x^\ell,y^{\ell'}) \not \in \Eq(\sD)$, $\lambda_\ell(x) \neq \lambda_{\ell'}(y)$.
\end{defn}
Intuitively, this definition says that $\sD$ is free if any two variables in $\sD$ which are not ``forced'' to have the same value are assigned different values.

\subsection{Canonical Programs}

Fix a $\tau$-structure $\bB$ and $j \leq k$. Let $Q_1,\dots,Q_n$ be all possible at most $j$-ary relations over $B$. The \emph{canonical linear $(j,k)$-Datalog program for $\bB$ ($(j,k)$-$\CL(\bB)$)} contains an IDB $I_m$ of the same arity as $Q_m$ for each $m \in [n]$. The rule $I_c(\mb{x}) \leftarrow I_d(\mb{y}) \wedge \bar{E}(\mb{z})$ belongs to the canonical program if it contains at most $k$ variables, and the implication $Q_c(\mb{x}) \leftarrow Q_d(\mb{y}) \wedge \bar{E}(\mb{z})$ is true for all possible instantiation of the variables to elements of $B$. The goal predicate of this program is the $0$-ary IDB $I_g$, where $Q_g = \emptyset$.

The \emph{canonical symmetric $(j,k)$-Datalog program for $\bB$ ($(j,k)$-$\CS(\bB)$)} has the same definition as $(j,k)$-$\CL(\bB)$, except that it has less rules due to the following additional restriction. If $I_c(\mb{x}) \leftarrow I_d(\mb{y}) \wedge \bar{E}(\mb{z})$ is in the program, then both $Q_c(\mb{x}) \leftarrow Q_d(\mb{y}) \wedge \bar{E}(\mb{z})$ and $Q_d(\mb{y}) \leftarrow Q_c(\mb{x}) \wedge \bar{E}(\mb{z})$ must hold for all possible instantiation of the variables to elements of $B$. The program $(j,k)$-$\CS(\bB)$ is obviously symmetric. When it is clear from the context, we write $\CL(\bB)$ and $\CS(\bB)$ instead of $(j,k)$-$\CL(\bB)$ and $(j,k)$-$\CS(\bB)$, respectively.

\subsection{Defining CSPs}\label{Defining_CSPs}

The following discussion applies not just to Datalog but also to its symmetric and linear fragments. It is easy to see that the class of structures accepted by a Datalog program is homomorphism-closed, and therefore it is not possible to define $\CSP(\bB)$ in Datalog. However, $\coCSP(\bB)$ is closed under homomorphisms, and in fact, it is often possible to define $\coCSP(\bB)$ in Datalog.

The following definition is key.
\begin{defn}[Obstruction Set]
A set $\cO$ of $\tau$-structures is called an \emph{obstruction set} for $\bB$, if for any $\tau$-structure $\bA$, $\bA \not \rightarrow \bB$ if and only if there exists $\bS \in \cO$ such that $\bS \rightarrow \bA$.
\end{defn}
In other words, an obstruction set defines $\coCSP(\bB)$ implicitly as $\bA \in \coCSP(\bB)$ if and only if there exists $\bS \in \cO$ such that $\bS \rightarrow \bA$. If $\cO$ above can be chosen to have property $X$, then we say that $\bB$ has $X$-duality. In the next section we show that $\coCSP(\bB)$ is definable in symmetric Datalog if and only if $\bB$ has \emph{symmetric bounded pathwidth} duality.


\section{On CSPs in symmetric Datalog}

\subsection{Definitions}\label{L_defs}

An \emph{oriented path} is a digraph obtained by orienting the edges of an undirected path. In other words, an oriented path has vertices $v_0,\dots,v_{q+1}$ and edges $e_0,\dots,e_q$, where $e_i$ is either $(v_i,v_{i+1})$, or $(v_{i+1},v_i)$. The \emph{length} of an oriented path is the number of edges it contains. We call $(v_i,v_{i+1})$ a \emph{forward edge} and $(v_{i+1},v_i)$ a \emph{backward edge}. Oriented paths can be thought of as relational structures over the vocabulary $\{E^2\}$, so we denote them with boldface letters.

For an oriented path $\bP$, we can find a mapping $\lev : P \rightarrow \{0,1,2,\dots\}$ such that $\lev(b) = \lev(a) + 1$ whenever $(a,b)$ is an edge of $\bP$. Clearly, there is a unique such mapping with the smallest possible values. The \emph{level of an edge} $(a,b)$ of $\bP$ is $\lev(a)$, i.e., the level of the starting vertex of $(a,b)$. The \emph{$\hei(\bP)$} of an oriented path $\bP$ is $\max_{a \in P} \lev(a)$. Let $\bP$ be an oriented path that has a vertex $u$ with indegree $0$ and outdegree $1$, and a vertex $v$ with indegree $1$ and outdegree $0$. We say that $\bP$ is \emph{minimal} if $u$ is in the bottommost level and $v$ is in the topmost level, and there are no other vertices of $\bP$ in the bottommost or the topmost levels.


A \emph{zigzag operator $\xi$} takes a $(j,k)$-path representation $\mathscr{S} = (\bS_0,\dots,\bS_{n-1})$ of a $(j,k)$-path $\bS$ and a minimal oriented path $\bP = e_0,\dots,e_q$ such that $\hei(\bP) = n$, and it returns another $(j,k)$-path $\xi(\mathscr{S},\bP)$. Intuitively, $\xi(\mathscr{S},\bP)$ is the $(j,k)$-path $\bS$ ``modulated'' by $\bP$ such that the forward and backward edges $e_i$ of $\bP$ are mimicked in $\xi(\mathscr{S},\bP)$ by ``forward and backward'' copies of $\bS_{\lev(e_i)}$. Before the formal definition, it could help the reader to look at the right side of Figure~\ref{zigzag}, where the oriented path used to modulate the $(j,k)$-path over the vocabulary $E^2$ (i.e., digraphs) with representation $(\bS_0,\bS_1,\bS_2)$ is $\bP$ on the left side. The left side is a more abstract example, and the reader might find it useful after reading the definition.

We inductively define the $(j,k)$-path $\xi(\mathscr{S},\bP)$ as $(\bS_{e_0}, \bS_{e_1},\dots, \bS_{e_q})$ together with a sequence of isomorphisms $\varphi_{e_0}, \varphi_{e_1}, \dots, \varphi_{e_q}$, where $\varphi_{e_i}$ is an isomorphism from $\bS_{e_i}$ to $\bS_{\lev(e_i)}$, $0 \leq i \leq q$. For the base case, we define $\bS_{e_0}$ to be an isomorphic copy of $\bS_0$, and $\varphi_{e_0}$ to be the isomorphism that maps $\bS_{e_0}$ back to $\bS_0$. Assume inductively that $\bS_{e_0}, \dots, \bS_{e_{i-1}}$ and $\varphi_{e_0}, \dots, \varphi_{e_{i-1}}$ are already defined. Let $\bS_{e_i}'$ be an isomorphic copy of $\bS_{\lev(e_i)}$ with domain disjoint from \mbox{$S_{e_0} \cup \dots \cup S_{e_{i-1}}$}, and fix $\varphi_{e_i}'$ to be the isomorphism that maps back $S_{e_i}'$ to $S_{\lev(e_i)}$. We ``glue'' $\bS_{e_i}'$ to $\bS_{e_{i-1}}$ by renaming some elements of $\bS_{e_i}'$ to elements of $\bS_{e_{i-1}}$. To facilitate understanding, we can think of the already constructed structures $\bS_{e_0}, \dots, \bS_{e_{i-1}}$ as labels of the edges $e_0,\dots,e_{i-1}$ of $\bP$, respectively, and we want to determine $\bS_{e_i}$, the label of the next edge. The connection between $\bS_{e_{i-1}}$ and $\bS_{e_i}$ will be defined such that $\bS_{e_{i-1}}$ and $\bS_{e_i}$ ``mimic'' the orientation of the edges $e_{i-1}$ and $e_i$.

We resume our formal definition. Set $\ell = \lev(e_i)$, and let $\ell' = \ell - 1$ if $e_i$ is a forward edge, and $\ell' = \ell + 1$ if $e_i$ is a backward edge. If an element $x \in S_{e_i}'$ and an element $y \in S_{e_{i-1}}$ are both copies of the same element $a \in S_{\ell} \cap S_{\ell'}$, then rename $x$ to $y$ in $S_{e_i}'$. After all such elements are renamed, $\bS_{e_i}'$ becomes $\bS_{e_i}$. That is, for all $a \in S_{\ell} \cap S_{\ell'}$, rename $\varphi_{e_i}^{\prime -1}(a)$ in $\bS_{e_i}'$ to $\varphi_{e_{i-1}}^{-1}(a)$ to obtain $\bS_{e_i}$.

We define the isomorphism $\varphi_{e_i}$ from $\bS_{e_i}$ to $\bS_{\lev(e_i)}$ as:
\[
\varphi_{e_i}(x) = \begin{cases}
\varphi_{e_i}'(x) & \text{if $x \in S_{e_i}$ and $x \not \in S_{e_{i-1}}$}\\
\varphi_{e_{i-1}}(x) & \text{if $x \in S_{e_i} \cap S_{e_{i-1}}$}.
\end{cases}
\]

\begin{figure}[h!]
\begin{center}
\includegraphics[scale=0.99]{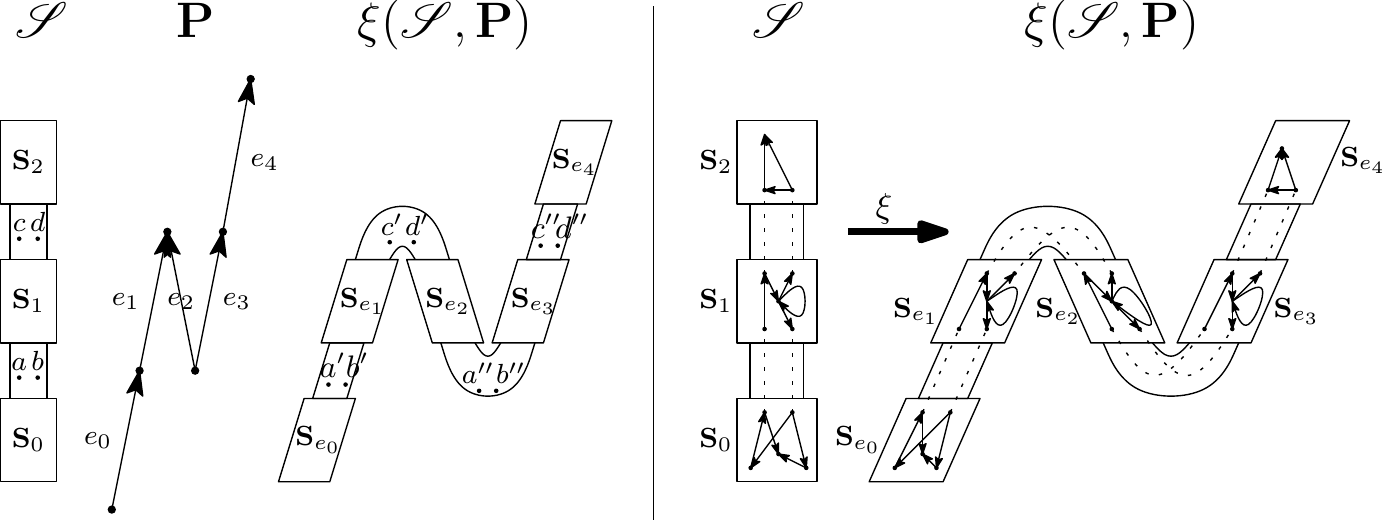}
 \caption{\emph{Left:} Applying a zigzag operator to the $(j,k)$-path $\bS$ with the $(j,k)$-representation $\mathscr{S} = (\bS_0,\bS_1,\bS_2)$. Suppose that $S_0 \cap S_1 = \{a,b\}$ and $S_1 \cap S_2 = \{c,d\}$. We demonstrate how $\bS_{e_0}$ and $\bS_{e_2}$ are obtained. \emph{$\bS_{e_0}$} is a disjoint copy of $\bS_0$ (and the copy of $a$ and $b$ in $\bS_{e_0}$ are $a'$ and $b'$, respectively). To obtain $\bS_{e_2}$, first make a disjoint copy $\bS_{e_2}'$ of $\bS_{\lev(e_2)} = \bS_1$. Set $\ell = \lev(e_2) = 1$. Since $e_1$ is a forward edge and $e_2$ is a backward edge, $\ell' = \ell + 1 = 2$. Therefore to ``glue''  $\bS_{e_2}'$ to $\bS_{e_1}$, we need to look at $S_{\ell} \cap S_{\ell'} = \{c,d\}$. Assume that the copy of $c$ and $d$ in $\bS_{e_1}$ are $c'$ and $d'$, respectively. Furthermore, assume that the copy of $c$ and $d$ in $\bS_{e_2}'$ are $\tilde{c}$ and $\tilde{d}$, respectively. To obtain $\bS_{e_2}$, we rename $\tilde{c}$ to $c'$, and $\tilde{d}$ to $d'$ in $\bS_{e_2}'$. \emph{Right:} A specific example when $\bS_0, \bS_1, \bS_2$ are the digraphs in the boxes. The dashed lines indicate identification of vertices. The level of $(\bS_{e_2}, \bS_{e_3})$, for example, is $0$ since $e_2$ and $e_3$ share a vertex at vertex level $1$.} \label{zigzag}
\end{center}
\end{figure}

\subsection{Two Dualities for Symmetric Datalog}

The two main theorems (Theorems~\ref{sym_obs_thm}~and~\ref{piece_sym_obs_thm}) of this section can be combined to obtain the equivalence of the statements \eqref{m1}, \eqref{m3} and \eqref{m4} in Theorem~\ref{main_duality_thm} below. The proof of the implication \eqref{m1} $\rightarrow$ \eqref{m2} is a direct adaptation of the proof of the result from \cite{Feder/Vardi:1999:Computational} that if $\coCSP(\bB)$ is defined by a $(j,k)$-Datalog program, then it is also defined by the canonical $(j,k)$-Datalog program (see also~\cite{Dalmau/Larose:2008:Maltsev}). Note that \eqref{m1} $\rightarrow$  \eqref{m2} is also obvious from the proof of Theorem~\ref{sym_obs_thm} below.
\begin{thm}\label{main_duality_thm}
For a finite structure $\bB$, TFAE:
\begin{enumerate}
\item There is a symmetric Datalog program that defines $\coCSP(\bB)$;\label{m1}
\item The canonical symmetric $(j,k)$-Datalog program defines $\coCSP(\bB)$;\label{m2}
\item $\bB$ has symmetric bounded pathwidth duality (for some parameters);\label{m3}
\item $\bB$ has piecewise symmetric bounded pathwidth duality (for some parameters).\label{m4}
\end{enumerate}
\end{thm}

\subsubsection{Symmetric Bounded Pathwidth Duality}

\begin{defn}[$(j,k)$-symmetric]\label{symmetric}
Assume that $\cO$ is a set of $(j,k)$-paths. Suppose furthermore that a $(j,k)$-path representation can be fixed for each structure in $\cO$ such that the following holds. For every $\bS \in \cO$ with representation $\sS$ of some length $n$, and every minimal oriented path $\bP$ of height $n$, it holds that $\xi(\sS,\bP) \in \cO$. Then $\cO$ is said to be \emph{$(j,k)$-symmetric}.
\end{defn}

\begin{defn}[SBPD]
A structure $\bB$ has \emph{$(j,k)$-symmetric bounded pathwidth duality ($(j,k)$-SBPD)} if there is an obstruction set $\cO$ for $\bB$ that consists of $(j,k)$-paths, and in addition, $\cO$ is $(j,k)$-symmetric.
\end{defn}

The following is our main duality theorem for symmetric Datalog:
\begin{thm}\label{sym_obs_thm}
For a finite structure $\bB$, $\coCSP(\bB)$ can be defined by a symmetric $(j,k)$-Datalog program if and only if $\bB$ has $(j,k)$-SBPD.
\end{thm}
We will use Lemma~\ref{easy_obs} in the proof of Theorem~\ref{sym_obs_thm}. Lemma~\ref{easy_obs} can be proved using the standard canonical Datalog program argument. Lemma~\ref{main_symmetric_pathwidth} is also used in the proof of Theorem~\ref{sym_obs_thm} and it is the main technical lemma of the section.
\begin{lem}\label{easy_obs}
 If $\CS(\bB)$ accepts a structure $\bA$, then $\bA \not \rightarrow \bB$.
\end{lem}
\begin{proof}[Proof]
Structure $\bB$ is not accepted by $\CS(\bB)$ because a derivation could be translated into a valid chain of implications, which is not possible by the definition of $\CS(\bB)$. If $\CS(\bB)$ accepts $\bA$ and $\bA \rightarrow \bB$, then $\CS(\bB)$ accepts $\bB$, a contradiction.
\end{proof}

\begin{lem}\label{main_symmetric_pathwidth}
For any $\tau$-structures $\bA$ and $\bB$, if there exists a structure $\bS$ with a $(j,k)$-path representation $\mathscr{S}$ of some length $n$ such that $\bS \rightarrow \bA$, and for any minimal oriented path $\bP$ of height $n$, it holds that $\xi(\mathscr{S}, \bP) \not \rightarrow \bB$, then $(j,k)$-$\CS(\bB)$ accepts $\bA$.
\end{lem}
To prove Lemma~\ref{main_symmetric_pathwidth} we need to define an additional concept related to the zigzag operator. Once the $(j,k)$-path $\xi(\mathscr{S},\bP) = (\bS_{e_0},\dots,\bS_{e_q})$ is defined, where $\bP$ is the path $e_0,\dots,e_q$, each pair $(\bS_{e_i}, \bS_{e_{i+1}})$, $\forall i \in \{0,\dots,q-1\}$ is assigned a \emph{level}: $\lev(\bS_{e_i}, \bS_{e_{i+1}})$ is the level of the vertex $v$ minus $1$, where $v$ is the vertex that $e_i$ and $e_{i+1}$ share (see Figure~\ref{zigzag}).

\begin{proof}[Proof of Lemma~\ref{main_symmetric_pathwidth}]
For the rest of this proof, let $\cC\cS$ denote $(j,k)$-$\CS(\bB)$, and $\cC\cL$ denote $(j,k)$-$\CL(\bB)$. If program $\cC\cS$ accepts structure $\bS$ then because $\bS \rightarrow \bA$, $\cC\cS$ also accepts $\bA$. So it is sufficient to show that program $\cC\cS$ accepts structure $\bS$.

First we specify how to associate a $\mathcal{CL}$-derivation with $\xi(\sS, \bP)$, where $\bP$ is a minimal oriented path of height $n$. Assume that $\xi(\sS, \bP) = \bS_0 \cup \dots \cup \bS_q$. For each \mbox{$i \in \{0,\dots,q-1\}$}, fix an arbitrary order on the elements of $S_i \cap S_{i+1}$. Assume that $|S_i \cap S_{i+1}| = j' (\leq j)$, and define the $j'$-tuple $\mathbf{s}_i$ such that $\mathbf{s}_i[\ell]$ is the $\ell$-th element of $S_i \cap S_{i+1}$. We define $\mathbf{s}_q$ to be the empty tuple. It is good to keep in mind that later, $\mathbf{s}_i$ will be associated with the IDB $J_i$.

The derivation will be $\mathscr{D}_{\cC\cL}(\xi(\sS,\bP)) = (\rho_0,\lambda_0), \dots, (\rho_q,\lambda_q)$. We specify $\rho_i$ as
\begin{align*}
J_i(\mathbf{x}_i) &\leftarrow J_{i-1}(\mathbf{x}_{i-1}) \wedge \bar{E}(\mathbf{y}_i) & J_0(\mathbf{x}_0) &\leftarrow \bar{E}(\mathbf{y}_0)\\
\text{if } i &\in [q] & \text{if } i &= 0.
\end{align*}
We begin with describing the EDBs of a rule $\rho_i$ together with their variables. Assume that $S_i = \{d_1,\dots,d_t\}$, and observe that $t \leq k$. The variables of $\rho_i$ are $v_1,\dots,v_t$. For every $R \in \tau$, and every tuple $(d_{f(1)},\dots,d_{f(r)}) \in R^{\bS_i}$, where $r = \ar(R)$, $R(v_{f(1)},\dots,v_{f(r)})$ is an EDB of $\rho_i$.

We describe the variables of the IDBs $J_{i-1}$ and $J_i$. Assume that $\mathbf{s}_{i-1} = (d_{g(1)},\dots,d_{g(j_1)})$ and $\mathbf{s}_i = (d_{h(1)},\dots,d_{h(j_2)})$. Then the IDB in the body of $\rho_i$ together with its variables is $J_{i-1}(v_{g(1)},\dots,v_{g(j_1)})$, and the head IDB together with its variables is $J_i(v_{h(1)},\dots,v_{h(j_2)})$. The function $\lambda_i$ simply assigns the value $d_g$ to the variable $v_g$, $\forall g \in [t]$.

It remains to specify the IDBs, i.e., which IDBs of $\cC\cL$ the $J_i$-s correspond to. For each $i \in \{0,\dots,q\}$, $I_i$ denotes $I_{M_i^{\bP}}$, where $M_i^{\bP}$ is a subset of $B^{j'}$ for some $j' \leq j$. We define the sequence $M_0^{\bP},M_1^{\bP},\dots,M_q^{\bP}$ inductively. To define $M_0^{\bP}$, consider the nonrecursive rule $J_0(\mathbf{x}_0) \leftarrow \bar{E}(\mathbf{y}_0)$. Assume that the arity of $J_0$ is $j'$, and that $\mathbf{y}_0$ contains $k'$ variables. (Note that the variables in $\mathbf{x}_0$ and $\mathbf{y}_0$ are not necessarily disjoint.) For all possible functions $\alpha : \mathbf{x}_0[1],\dots,\mathbf{x}_0[j'],\mathbf{y}_0[1],\dots, \mathbf{y}_0[k'] \rightarrow B$ such that the conjunction of EDBs $\bar{E}(\alpha(\mathbf{y}_0[1]),\dots,\alpha(\mathbf{y}_0[k']))$ is true, place the tuple $(\alpha(\mathbf{x}_0[1]),\dots,\alpha(\mathbf{x}_0[j']))$ into $M_0^{\bP}$.

Assume that $M_{i-1}^{\bP}$ is already defined. Then similarly to the base case, for each possible instantiation $\alpha$ of the variables of $\rho_i$ over $B$ with the restriction that $\alpha(\mb{x}_{i-1}) \in M_{i-1}^{\bP}$, if the conjunction of EDBs of $\rho_i$ is true, then add the tuple $\alpha(\mb{x}_i)$ to $M_i^{\bP}$. It is not difficult to see that if $M_q^{\bP} \neq \emptyset$, then we can construct a homomorphism from $\xi(\mathscr{P},\bP)$ to $\bB$ which would be a contradiction.

For each $i \in \{0,\dots,q-1\}$, assume that $(\bS_i,\bS_{i+1})$ has level $\ell_i$. Then we say that the IDB $J_i$ has \emph{level} $\ell_i$ and we write $\lev(J_i) = \ell_i$.

We proceed to construct a $\cC\cS$-derivation $\mathscr{D}_{\cC\cS}(\bS)$ for $\bS$. Let $\bQ$ be a directed path of height $n$. We construct $\mathscr{D}_{\cC\cS}(\bS)$ just like we would construct $\mathscr{D}_{\cC\cL}(\xi(\mathscr{S},\bQ))$ above, except that we will define the subscripts of the IDBs, $M_0^\bQ,\dots,M_{n-1}^\bQ$, differently, so that every rule of the resulting derivation belongs to $\cC\cS$. From now on we write $M_0,\dots,M_{n-1}$ instead of $M_0^\bQ,\dots,M_{n-1}^\bQ$.

To define $M_0,\dots,M_{n-1}$, let $\bP_0, \bP_1, \dots$ be an enumeration of all (finite) minimal oriented paths of height $n$. Intuitively, we will collect in $\cN^\ell_m$ all subscripts (recall that a subscript is a relation) of all those IDBs which have the same level $\ell$ in $\mathscr{D}_{\cL}(\xi(\mathscr{S},\bP_m))$. Formally, for each $\ell \in \{0,\dots,n-1\}$ define $\cN^\ell_m = \{M_t^{\bP_m} \; | \; \lev(J_t)=\ell\}$.
Then we collect the subscripts at a fixed level $\ell$ in $\cO_\ell$ over all derivations corresponding to $\bP_0,\bP_1,\dots$. Formally, for each $\ell \in \{0,\dots,n-1\}$, we define $\cO_\ell = \cN^\ell_0 \cup \cN^\ell_1,\dots$. We are ready to define $M_0,\dots,M_{n-1}$. For each $s \in \{0,\dots,n-1\}$, define
$M_s = \bigcup_{W \in \cO_s} W$.

It remains to show that every rule of the derivation we defined is in $\cS$ and that the last IDB is the goal IDB. If the last IDB is not the goal IDB of $\cS$, then $M_{n-1} \neq \emptyset$. By definition, it must be that for some minimal oriented path $\bP_m$ of height $n$ and length $q_m$, $\bM^{\bP_m}_{q_m-1} \neq \emptyset$ (note that the last IDB of $\mathscr{D}_{\cC\cL}(\xi(\mathscr{P}, \bP_m))$ has subscript $\bM^{\bP_m}_{q_m-1}$). As noted before, this would mean that $\xi(\mathscr{P},\bP_m) \rightarrow \bB$, a contradiction.

We show that each rule of $\mathscr{D}_{\cC\cS}(\bS)$ as defined above belongs to $\CS(\bB)$. Suppose $\mathscr{D}_{\cC\cS}(\bS)$ contains a rule $\rho$
\[
J_i(\mathbf{x}_i) \leftarrow J_{i-1}(\mathbf{x}_{i-1}) \wedge \bar{E}(\mathbf{y}_i)
\]
that is not in $\CS(\bB)$. By definition, there cannot be an instantiation $\alpha$ of variables of $\rho$ to elements of $B$ such that $\alpha(\mathbf{x}_{i-1}) \in M_{i-1}$, the conjunction of EDBs holds, but $\alpha(\mathbf{x}_i) \not \in M_i$. Assume then that there is an $\alpha$ such that $\alpha(\mathbf{x}_i) \in M_i$, the conjunction of EDBs holds, but $\alpha(\mathbf{x}_{i-1}) \not \in M_{i-1}$. It is also not difficult to see that this is not possible because we used \emph{all} minimal oriented paths in the construction of $\mathscr{D}_{\cC\cS}(\bS)$.
\end{proof}

\begin{proof}[Proof of Theorem~\ref{sym_obs_thm}]
If $\CSP(\bB)$ is defined by a symmetric $(j,k)$-Datalog program $\cP$, then using the symmetric property of $\cP$, it is laborious but straightforward to show that
\[\cO = \bigcup_{\substack{\text{$\sD$ is a free}\\ \text{derivation of $\cP$}}} \{\Ex(\sD)\}\]
is a $(j,k)$-symmetric obstruction set for $\bB$.

For the converse, assume that $\bB$ has $(j,k)$-SBPD. Let $\cO$ be a symmetric obstruction set of width $(j,k)$ (i.e., the path decomposition of every structure in $\cO$ has width $(j,k)$) for $\bB$. We claim that $(j,k)$-$\CS(\bB)$ defines $\CSP(\bB)$. Assume that $\bA \rightarrow \bB$. Then by Lemma~\ref{easy_obs}, $(j,k)$-$\CS(\bB)$ does not accept $\bA$. Suppose now that $\bA \not \rightarrow \bB$. Then by assumption, there exists a $(j,k)$-path $\bS \in \cO$ with a representation $\mathscr{S}$ of length $n$ such that $\bS \rightarrow \bA$. Furthermore, since $\cO$ is symmetric, for any minimal oriented path $\bP$ of height $n$, $\xi(\mathscr{S},\bP) \not \rightarrow \bB$. It follows from Lemma~\ref{main_symmetric_pathwidth} that $\CS(\bB)$ accepts $\bA$.
\end{proof}


\subsubsection{Piecewise Symmetric Bounded Pathwidth Duality}

Piecewise symmetric bounded pathwidth duality (PSBPD) for symmetric Datalog is less stringent than SBPD; however, the price is larger program width. Although the following definitions might seem technical, the general idea is simple: a piecewise symmetric obstruction set $\cO$ does not need to contain all $(j,k)$-paths obtained by ``zigzagging'' $(j,k)$-paths in $\cO$ in all possible ways. It is sufficient to zigzag a $(j,k)$-path $\bS$ using only oriented paths which ``avoid'' certain segments of $\bS$: some constants $c$ and $d$ are fixed for $\cO$, and there are at most $c$ fixed segments of $\bS$ that are avoided by the zigzag operator, each of size at most $d$. We give the formal definitions.

\begin{defn}[$(c,d)$-filter]
Let $\bS$ be a $(j,k)$-path with a representation $\sS = \bS_0 , \dots , \bS_{n-1}$. A \emph{$(c,d)$-filter} $\mathscr{F}$ for $\sS$ is a set of intervals $\{[s_1,t_1],[s_2,t_2],\dots,[s_{c'}, t_{c'}]\}$
such that
\begin{itemize}
\item $c' \leq c$; $0 \leq s_1$; $t_{c'} \leq n-1$; $s_i \leq t_i, \forall i \in[c']$; and $t_\ell + 2 \leq s_{\ell+1}, \forall \ell \in [c'-1]$;
\item $|\bigcup_{i \in [s_\ell, t_\ell]} S_i| \leq d, \forall \ell \in [c']$.
\end{itemize}
Elements of $\mathscr{F}$ are called \emph{delimiters}. An oriented path $\bP$ of height $n$ \emph{obeys} a $(c,d)$-filter $\sF$ if for any delimiter $[s_i,t_i] \in \sF$, the set of edges $e$ of $\bP$ such that $s_i \leq \lev(e) \leq t_i$ form a (single) directed path. A demonstration is given in Figure~\ref{filter}.
\end{defn}
\begin{figure}[h!]
\begin{center}
\includegraphics[scale=1]{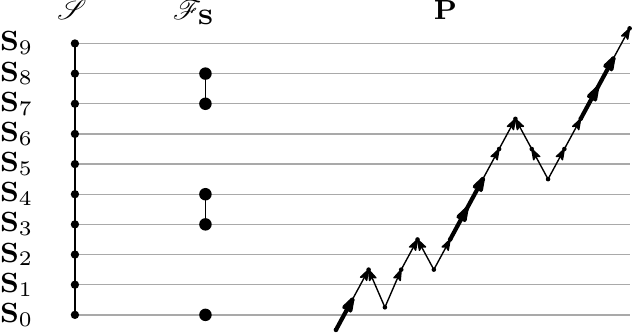}
\caption{$\mathscr{S}$ is a $(j,k)$-path representation of $\bS$. $\mathscr{F}_{\bS}$ is the $(3,2k)$-filter $\{[0,0],[3,4],[7,8]\}$ for $\mathscr{S}$. $\bP$ is an oriented path that obeys the filter. For example, observe that the edges at levels 3 and 4 form a directed subpath, and that ``zigzagging'' happens only at those parts of $\bP$ that do not fall into the intervals of the filter.}\label{filter}
\end{center}
\vspace{-0.5cm}
\end{figure}

\begin{defn}[Piecewise Symmetric]\label{piecewise_sym}
Assume that $\cO$ is a set of $(j,k)$-paths, and $c$ and $d$ are nonnegative integers. Suppose furthermore that for each $\bS \in \cO$, there is a $(j,k)$-path representation $\sS$, and a $(c,d)$-filter $\sF_{\bS}$ such that the following holds. For every $\bS \in \cO$ of some length $n$, and every minimal oriented path $\bP$ of height $n$ that obeys the filter $\sF_{\bS}$, it holds that $\xi(\sS,\bP) \in \cO$. Then $\cO$ is \emph{$(j,k,c,d)$-piecewise symmetric}.
\end{defn}
Roughly speaking, an oriented path $\bP$ is allowed to modulate only those segments of $\sS$ which do not correspond to any delimiters in $\sF_{\bS}$. Compare Definition~\ref{piecewise_sym} with Definition~\ref{symmetric}, and observe that the only difference is that in the piecewise case, the oriented paths must be of a restricted form. Therefore a set that is $(j,k)$-symmetric is also $(j,k,c,d)$-piecewise symmetric for any $c$ and $d$. We simply associate the empty $(c,d)$-filter with each structure.
\begin{defn}[PSBPD]
A structure $\bB$ has \emph{$(j,k,c,d)$-piecewise symmetric bounded pathwidth duality ($(j,k,c,d)$-PSBPD)} if there is an obstruction set $\cO$ for $\bB$ that consists of $(j,k)$-paths, and in addition, $\cO$ is $(j,k,c,d)$-piecewise symmetric.
\end{defn}
\begin{thm}\label{piece_sym_obs_thm}
For a finite structure $\bB$, $\bB$ has SBPD (for some parameters) if and only if $\bB$ has PSBPD (for some parameters).
\end{thm}

We need the corollary of the following lemma in the proof of the above theorem.
\begin{lem}\label{combine_ori}
Let $\bP$ be a minimal oriented path $e_0,\dots,e_{n-1}$ with the $(1,2)$-path representation $\sP = (e_0,\dots,e_{n-1})$, where we think of $e_i$ as a structure with two domain elements and a binary relation that contains the tuple $e_i$. Let $\bQ$ be a minimal oriented path $f_0,\dots,f_m$ with $n$ edge levels. Then the oriented path $\xi(\sP, \bQ)$ is minimal and has the same height as $\bP$.
\end{lem}
\begin{proof}
It is obvious that $\xi(\sP, \bQ)$ is an oriented path. Furthermore the map that assigns every vertex of $\xi(\sP, \bQ)$ to its original in $\bP$ is a homomorphism. It is easy to check that this homomorphism maps the edges of $\xi(\sP, \bQ)$ back to their originals and the level of an edge in $\xi(\sP, \bQ)$ is the same as the level of the original of that edge. Checking the minimality of $\xi(\sP, \bQ)$ is also straightforward.
\end{proof}

\begin{coro}\label{just_once}
Let $\cO$ be a set of $(j,k)$-paths, where a $(j,k)$-representation is fixed for each path. Let $\cO'$ be the set that contains all $(j,k)$-paths that can be obtained from a $(j,k)$-path in $\cO$ by applying some zigzag operator. Then $\cO'$ is $(j,k)$-symmetric.
\end{coro}
\noindent \emph{Remark:} A similar statement holds in the piecewise symmetric case.
\begin{proof}
Let $\bS'$ be an element of $\cO'$. If we can show that applying an arbitrary zigzag operator to $\bS'$ yields a $(j,k)$-path in $\cO'$, then we are clearly done. So assume that $\bS'$ was obtained from $\bS \in \cO$ by applying a zigzag operator. The $(j,k)$-path $\bS'$ inherits the $(j,k)$-representation of $\bS$ in a natural way. Then we apply any zigzag operator to $\bS'$ to obtain $\bS''$, and we need to show that $\bS''$ is in $\cO'$.

We get from $\bS$ to $\bS'$ using a zigzag operator and from $\bS'$ to $\bS''$ another zigzag operator. Using Lemma~\ref{just_once}, we can see that we can replace these two zigzag operators by a single one to obtain $\bS''$ from $\bS$ directly.
\end{proof}

\begin{proof}[Proof of Theorem~\ref{piece_sym_obs_thm}]
Let $\cO$ be a $(j,k)$-symmetric obstruction set for $\bB$. As observed above, for any $c$ and $d$, $\cO$ is also $(j,k,c,d)$-piecewise symmetric.

For the converse, let $\cO$ be a $(j,k,c,d)$-piecewise symmetric obstruction set. Our goal is to construct a $(j',k')$-symmetric obstruction set $\cO_{sym}$ for $\bB$ as follows. For each structure $\bS \in \cO$, let $\sS = \bS_0 \cup \dots \cup \bS_{n-1}$ be the corresponding $(j,k)$-path representation. Using the filter for $\bS$, we ``regroup'' $\bS_0,\dots,\bS_{n-1}$ to obtain $(j',k')$-path representation $\sS' = \bT_0 \cup \dots \cup \bT_m$ of $\bS$. We add each $\bS$ together with its new representation to $\cO_{sym}$, and also add every structure that is needed to ensure that $\cO_{sym}$ is symmetric. Finally, we show that $\cO_{sym}$ is a symmetric obstruction set for $\bB$. We begin with the regrouping procedure.

Let $\bS \in \cO$, $\sS = \bS_0 \cup \dots \cup \bS_{n-1}$ be the corresponding $(j,k)$-path representation, and $\{[s_1,t_1],[s_2,t_2],\dots,[s_{c'}, t_{c'}]\}$ be the $(c,d)$-filter $\sF_{\bS}$. The regrouping procedure is quite pictorial and it is demonstrated in Figure~\ref{regrouping}.
\begin{figure}[h!]
\begin{center}
\includegraphics[scale=1]{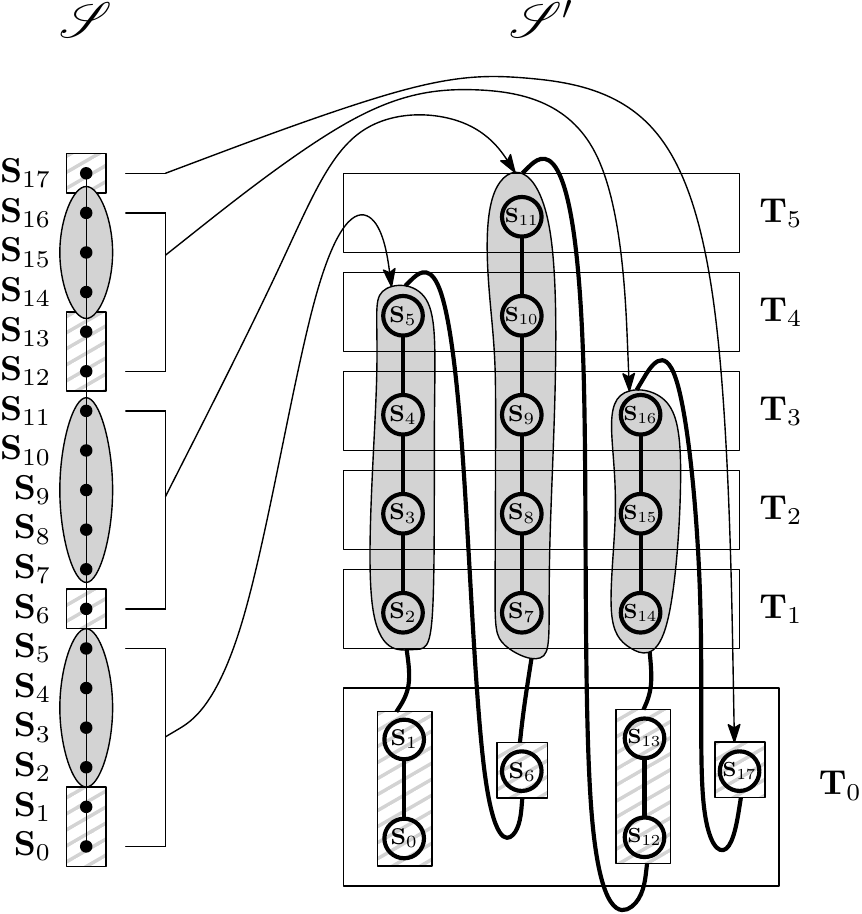}
\caption{An example regrouping for the proof of Theorem~\ref{piece_sym_obs_thm}. The filter $\cF_{\bS} = \{[0,1],[6],[12,13],[17]\}$. The structures corresponding to the filter are laying inside the rectangles with lines. The complement of the filter is $\bar{\cF}_{\bS} = \{[2,3,4],[7,8,9,10,11],[14,15,16]\}$. The structure corresponding to $\bar{\cF}_{\bS}$ lay in the gray ovals. The new $(j',k')$-path representation $\sS'$ of $\bS$ is on the right. Notice the following pattern: the segments of $\sS$ determined by $\sF_{\bS}$ are placed next to each other in $\sS'$.}\label{regrouping}
\end{center}
\end{figure}
We define
\[
\bT_0 = \bigcup_{\substack{\ell \in [a,b] : \\ [a,b] \in \sF_{\bS}}} \bS_\ell.
\]
This places all substructures in $\sS$ which correspond to delimiters of $\sF_{\bS}$ into one big initial structure. Note though that $|T_0| \leq c \cdot d$.
Define the \emph{complement} of $\sF_{\bS}$ as $\bar{\sF}_{\bS} =$
\[\{[0,s_1-1],[t_1+1,s_2-1],[t_2+1,s_3-1], \dots, [t_{c'},n-1]\},\]
and set
\[m = \max_{[a,b] \in \bar{\sF}_{\bS}} (b-a).\] Intuitively, $m$ is the length of the longest interval in $\sS$ between any two delimiters.

We define $\bT'_{\ell}$ as follows. For each interval $[a,b] \in \bar{\sF}_{\bS}$ take the $(\ell-1)$-th structure $\bS_{a+\ell-1}$ in that interval and define $\bT'_{\ell}$ to be the union of these structures. Formally, for every $\ell \in \{1, \dots, m\}$, set \[\bT_\ell' = \bigcup_{\substack{i = a + \ell - 1 \leq b: \\ [a,b] \in \bar{\sF}_{\bS}}} \bS_i.\]
Observe that $|T'_\ell| \leq k \cdot (c+1)$. We need to ensure property~\ref{inherit} in Definition~\ref{path_decomp}, so we need to place some additional elements into the domains of the $\bT_\ell'$.

Let $[x,y] \in \sF_{\bS}$ and $[z,w] \in \bar{\sF}_{\bS}$ be such that $z = y+1$. Then the set of elements $S_x \cup \dots \cup S_w$ is called a \emph{column}. (For the beginning and end of $\sS$ a column is defined in the natural ``truncated'' way.) Because $\sS$ is a $(j,k)$-path representation, it follows from the definition that the intersection of any pair of columns has size at most $j$. Let $C_1,\dots,C_r$ be an enumeration of all the columns. Set $D = \bigcup_{\ell \neq \ell'} C_\ell \cap C_{\ell'}$ and observe that $|D| \leq j \cdot \binom{r}{2}$. We add $D$ to the domain of $\bT_0$, and also to the domain of $\bT_i'$ to obtain $\bT_i$, $\forall i \in \{1,\dots,m\}$. It is straightforward to see that the new representation $\mathscr{T} = (\bT_0, \dots, \bT_m)$ satisfies property~\ref{inherit} of Definition~\ref{path_decomp}. Using the remarks about the sizes of the sets, we observe that $\mathscr{T}$ is a $(j',k')$-path decomposition of $\bS$, where $j'$ and $k'$ are functions of $j,k,c$ and $d$.

We place all structures $\bS \in \cO$ into $\cO_{sym}$ but we associate the new representation with $\bS$. For a structure $\bS \in \cO_{sym}$, we also apply all valid zigzag operators to $\bS$ (with respect to the new representation) and add all these structure to $\cO_{sym}$. By Lemma~\ref{just_once}, $\cO_{ps}$ is a $(j',k')$-symmetric set. We need to establish that $\cO_{ps}$ is an obstruction set. Because $\cO \subseteq \cO_{sym}$, it is sufficient to show that no structure in $\cO_{sym}$ maps to $\bB$. To do that we show that for any structure in $\cO_{ps}$, there is a structure in $\cO$ that homomorphically maps to it.

Giving a formal proof would lead to unnecessary notational complications and therefore we give an example that is easier to follow and straightforward to generalize. The example is represented in Figure~\ref{piecewise}. Let $\bS \in \cO_{ps}$ such that $\bS$ is also in $\cO$. Assume that the $(j',k')$-representation of $\bS$ in $\cO_{ps}$ is $\sT$. We consider $\xi(\sT, \bP)$ for some minimal oriented path and show how to find a minimal oriented path $\bQ$ such that $\xi(\sS, \bQ) \rightarrow \xi(\sT, \bP)$.
\begin{figure}
\begin{center}
\includegraphics[scale=0.75]{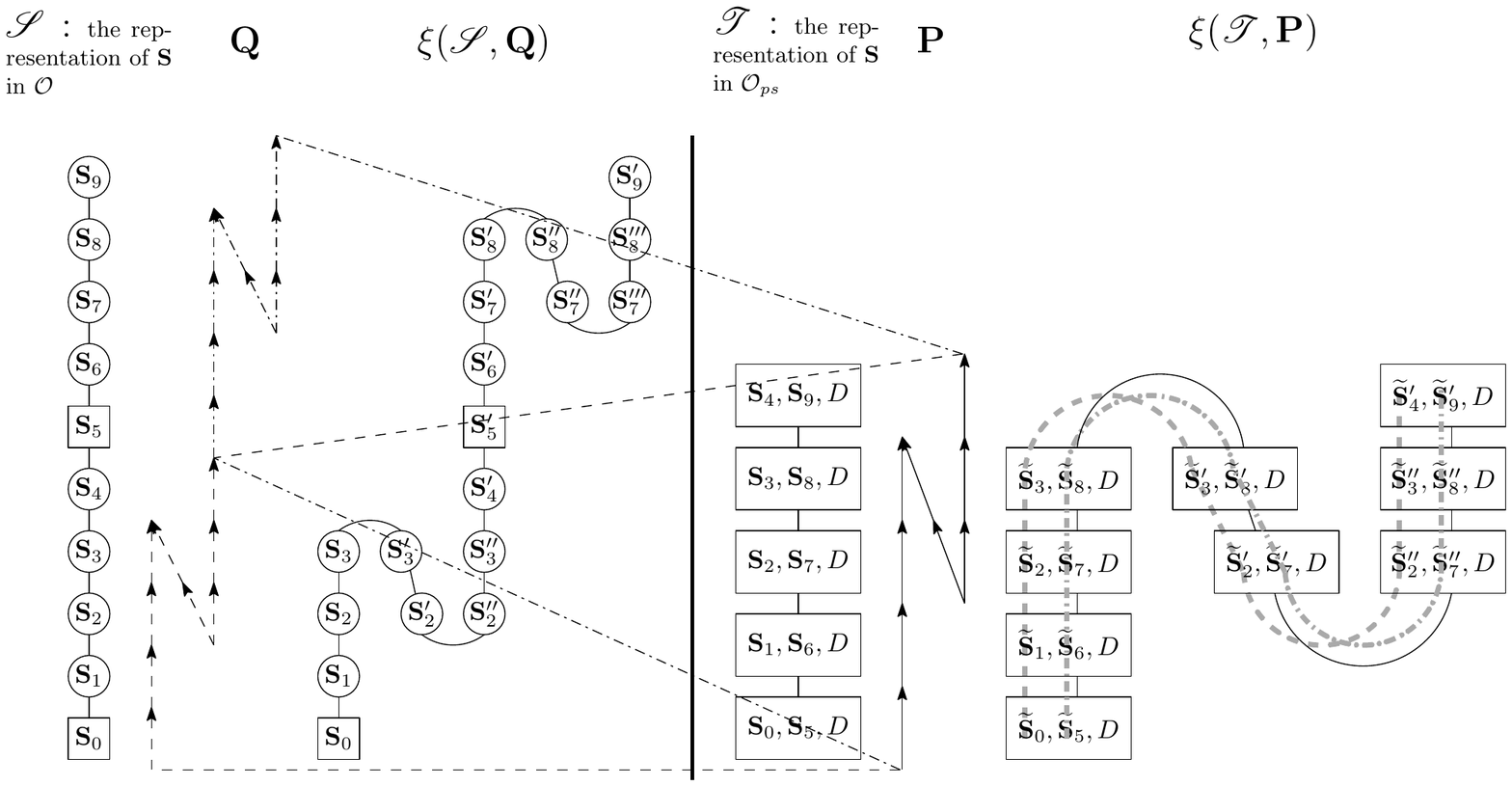}
\caption{Example in the proof of Theorem~\ref{piece_sym_obs_thm}. See the proof for details.}\label{piecewise}
\end{center}
\end{figure}
To construct $\bQ$, we make a copy of $\bP$ aligned with $\bS_0,\bS_1,\bS_2,\bS_3,\bS_4$ in $\sS$. This is represented by the dashed lines in Figure~\ref{piecewise}. We also make a copy of $\bP$ aligned with $\bS_5,\bS_6,\bS_7,\bS_8,\bS_9$. This is represented with the dash dotted lines. Note that the resulting minimal oriented path respects the delimiters, i.e., the zigzag operator will not ``zigzag'' $\bS_0$ and $\bS_5$. (In general, we never need to ``zigzag'' structures that were placed into $\bT_0$, i.e., the structures that correspond to the delimiters, because $\bP$ is minimal.)

In $\xi(\sT,\bP)$ we denote the copies of the $\bS_i$ with $\widetilde{\bS}_i$ and primed $\widetilde{\bS}_i$. Using the definition of the zigzag operator, it follows that the function $f$ that maps an element of $\bS_0 \cup \bS_1 \cup \bS_2 \cup \bS_3 \cup \bS_3' \cup \bS_2' \cup \bS_3'' \cup \bS_4'$ in $\xi(\sS,\bQ)$ to the corresponding element in $\widetilde{\bS}_0 \cup \widetilde{\bS}_1 \cup \widetilde{\bS}_2 \cup \widetilde{\bS}_3 \cup \widetilde{\bS}_3' \cup \widetilde{\bS}_2' \cup \widetilde{\bS}_3'' \cup \widetilde{\bS}_4'$ is a homomorphism. We similarly define a homomorphism $h$ from $\bS_5' \cup \bS_6' \cup \bS_7' \cup \bS_8' \cup \bS_8'' \cup \bS_7'' \cup \bS_7''' \cup \bS_8''' \cup \bS_9'$ in $\xi(\sS,\bQ)$ to $\widetilde{\bS}_5' \cup \widetilde{\bS}_6' \cup \widetilde{\bS}_7' \cup \widetilde{\bS}_8' \cup \widetilde{\bS}_8'' \cup \widetilde{\bS}_7'' \cup \widetilde{\bS}_7''' \cup \widetilde{\bS}_8''' \cup \widetilde{\bS}_9'$ in $\xi(\sT,\bP)$.

If we can make sure that if an element $x$ is in the domain of both $f$ and $h$, and both homomorphisms map $x$ to the same element then we have the desired homomorphism. Assume for example that the element $x$ appears in $\bS_2$ and also in $\bS_8'''$ in $\xi(\sS,\bQ)$, and suppose that $f(x)=y$ and $h(x)=y'$. Let the originals of $y$ and $y'$ be $z$ and $z'$ in $\sT$, respectively. We also identify $z$ and $z'$ in $\bS_2$ and $\bS_8$ in $\sS$. Observe that $x$ in $\bS_2$ in $\xi(\sS,\bQ)$ is a copy of $z$ and $x$ in $\bS_8'''$ in $\xi(\sS,\bQ)$ is a copy of $z'$. If $z \neq z'$ (in $\sS$) then $x$ could not appear both in $\bS_2$ and $\bS_8'''$ by the definition of the zigzag operator. Therefore $z = z'$, $z \in D$, and by definition, $z$ is in every bag of $\sT$. The elements $y$ and $y'$ are copies of $z$, and because $z$ appears in every ``bag'' of $\sT$, all copies of $z$ in $\xi(\sT,\bP)$ are identified to be the same element. In particular, $f(x) = y = y' = h(x)$.
\end{proof}

\subsection{Applications}

\subsubsection{Datalog + Maltsev $\Rightarrow$ Symmetric Datalog}
Using SBPD, we give a short and simple re-proof of the main result of \cite{Dalmau/Larose:2008:Maltsev}:
\begin{thm}[\cite{Dalmau/Larose:2008:Maltsev}]\label{MalDat}
Let $\bB$ be a finite core structure. If $\bB$ is invariant under a Maltsev operation and $\coCSP(\bB)$ is definable in Datalog, then $\coCSP(\bB)$ is definable in symmetric Datalog (and therefore $\CSP(\bB)$ is in $\ccL$ by \cite{Egri/Larose/Tesson:07:Symmetric}).
\end{thm}

We only need to show that if $\coCSP(\bB)$ is in linear Datalog and $\bB$ is preserved by a Maltsev operation, then $\coCSP(\bB)$ is in symmetric Datalog. The ``jump'' from Datalog to linear Datalog essentially follows from already established results, as observed in \cite{Dalmau/Larose:2008:Maltsev}. For the sake of completeness, we give an approximate outline of the argument without being too technical.\footnote{The interested reader can consult Lemma~6 (originally in \cite{Pixley:63:Distributivity}) and Lemma~9 in \cite{Dalmau/Larose:2008:Maltsev}. For Lemma~9, note that if $\bB$ has a Maltsev polymorphism, then $\cV(\bbA(\bB))$ is congruence permutable, see \cite{Burris/et_al:1981:Course}.} If $\coCSP(\bB)$ is definable in Datalog and $\bB$ has a Maltsev polymorphism, then $\bB$ also has a majority polymorphism. If $\bB$ has a majority polymorphism, then $\coCSP(\bB)$ is definable in linear Datalog \cite{Dalmau/Krokhin:2008:Majority}. Hence, to re-prove Theorem~\ref{MalDat}, it is sufficient to prove Lemma~\ref{reproof}. Our proof relies on the notion of SBPD.

\begin{lem}\label{reproof}
If $\coCSP(\bB)$ is definable by a linear Datalog program and $\bB$ is invariant under a Maltsev operation $m$, then $\coCSP(\bB)$ is definable by a symmetric Datalog program.
\end{lem}

To get ready for the proof of Lemma~\ref{reproof}, we define an $N$-digraph of size $s$ as an oriented path that consists of $s$ forward edges, followed by $s$ backward edges, followed by another $s$ forward edges. Proposition~\ref{N_in_minimal} is easy to prove, and the Maltsev properties are used in Lemma~\ref{rect_lemma}.
\begin{prop}\label{N_in_minimal}
A minimal oriented path is either a directed path, or it contains a subpath which is an $N$-digraph.
\end{prop}
\begin{lem}\label{rect_lemma}
Let $\bB$ be a structure invariant under a Maltsev operation $m$, $\bS$ be a $(j,k)$-path with a $(j,k)$-representation $\mathscr{S} = (\bS_0, \dots, \bS_{n-1})$, and $\bP = e_0,\dots,e_q$ be a minimal oriented path of height $n$. If $\xi(\mathscr{S},\bP) \rightarrow \bB$, then $\bS \rightarrow \bB$.
\end{lem}
\begin{proof}
Using Proposition~\ref{N_in_minimal}, there is an index $t$ such that $\bQ = e_t,e_{t+1},\dots,e_{t+(3s-1)}$ is an $N$-digraph of size $s$ in $\bP$. Assume that the first and last vertices of $\bQ$ are $v$ and $w$, respectively. Let $\bP'$ be the oriented path obtained from $\bP$ by removing $\bQ$, and adding a directed path $\bQ' = f_t,f_{t+1},\dots,f_{t+(s-1)}$ of length $s$ from $v$ to $w$. We claim that there is a homomorphism $\gamma$ from $\xi(\mathscr{S},\bP')$ to $\bB$. Once this is established, repeating the argument sufficiently many times clearly yields that $\bS \rightarrow \bB$.

Let $\xi(\mathscr{S},\bP) = (\bS_{e_0},\dots,\bS_{e_q})$, and $\varphi_{e_0},\dots,\varphi_{e_q}$ be the corresponding isomorphisms (recall the zigzag operator definition in Section~\ref{L_defs}). Similarly, let $\xi(\mathscr{S},\bP') = (\bS_{f_0},\dots,\bS_{f_{q-2s}})$, and $\psi_{f_0},\dots,\psi_{f_{q-2s}}$ be the corresponding isomorphisms. Because $\bS_{[e_0,e_{t-1}]}$ and $\bS_{[e_{t+3s},e_q]}$ are isomorphic to $\bS_{[f_0,f_{t-1}]}$ and $\bS_{[f_{t+s},f_{q-2s}]}$, respectively, $\gamma$ for elements in $S_{[f_0,f_{t-1}]} \cup S_{[f_{t+s},e_{q-2s}]}$ is defined in the natural way. It remains to define $\gamma$ for every $d \in S_{[f_t,f_{t+(s-1)}]}$.

Assume that $d \in S_{f_{t+\ell}}$ for some $\ell \in \{0,\dots,s-1\}$. Find the original of $d$ in $\bS$ and let it be $d_o$, i.e., $d_o = \psi_{f_{t+\ell}}(d)$. Then we find the three copies $d_1,d_2,d_3$ of $d_o$ in $\bS_{[f_t,f_{t+(3s-1)}]}$. That is, first we find the three edges $e_{\ell_1}, e_{\ell_2}, e_{\ell_3}$ of $\bQ$ which have the same level as $f_{t+\ell}$ (all levels are with respect to $\bP$ and $\bP'$). Then $d_i = \varphi_{e_{\ell_i}}^{-1}(d_o)$, $i \in [3]$. We define $\gamma(d) = m(d_1,d_2,d_3)$. By the Maltsev properties of $m$, $\gamma$ is well-defined. As $\bB$ is invariant under $m$, $\xi(\mathscr{S},\bP') \overset{\gamma}{\longrightarrow} \bB$.
\end{proof}

\begin{proof}[Proof of Lemma~\ref{reproof}]
If $\coCSP(\bB)$ can be defined by a linear $(j,k)$-Datalog program, then there is an obstruction set $\cO$ for $\bB$ in which every structure is a $(j,k)$-path by \cite{Dalmau:02:Constraint}. We construct a symmetric obstruction set $\cO_{sym}$ for $\bB$ as follows. For every $(j,k)$-path $\bS$ with a $(j,k)$-representation $\mathscr{S} = \bS_0, \dots, \bS_{n-1}$ in $\cO$ and for every minimal oriented path $\mathbf{P}$ of height $n$, place $\xi(\mathscr{S}, \bP)$ into $\cO_{sym}$. By Corollary~\ref{just_once}, $\cO_{sym}$ is $(j,k)$-symmetric.

Observe that $\cO \subseteq \cO_{sym}$, so it remains to show that no element of $\cO_{sym}$ maps to $\bB$. But if $\bT \in \cO_{sym}$, then $\bT = \xi(\sS,\bP)$ for some $\bS \in \cO$ and $\bP$. By Lemma~\ref{rect_lemma}, if $\xi(\sS,\bP) \rightarrow \bB$, then $\bS \rightarrow \bB$. This contradicts the assumption that $\cO$ is an obstruction set for $\bB$.
\end{proof}

\subsubsection{A class of oriented paths for which the CSP is in L, and a class for which the CSP is NL-complete}

In this section we define a class $\cC$ of oriented paths such that if $\bB \in \cC$ then $\coCSP(\bB)$ is in symmetric Datalog. Our strategy is to find an obstruction set $\cO$ for $\bB \in \cC$, and then to show that our obstruction set is piecewise symmetric. We need some notation.

We say that a directed path is \emph{forward} to mean that its first and last vertices are the vertices with indegree zero and outdegree zero, respectively. Let $\bP$ be an oriented path with first vertex $v$ and last vertex $w$. Then the \emph{reverse} of $\bP$, denoted by $\bar{\bP}$, is a copy of the oriented path $\bP$ in the reverse direction, i.e., the first vertex of $\bar{\bP}$ is a copy of $w$ and its last vertex is a copy of $v$. Let $\bQ$ be another oriented path. The \emph{concatenation} of $\bP$ and $\bQ$ is the oriented path $\bP \bQ$ in which the last vertex of $\bP$ is identified with the first vertex of $\bQ$. For a nonnegative integer $r$, $\bP^r$ denotes $\bP_1 \bP_2 \cdots \bP_r$, where the $\bP_\ell$ are disjoint copies of $\bP$. Given two vertices $v$ and $w$, we denote the presence of an edge from $v$ to $w$ with $v \rightarrow w$.
\begin{defn}[Wave]
If an oriented path $\bQ$ can be expressed as $\bE_1 (\bP \bar{\bP})^r \bP \bE_2$, where $\bE_i$ ($i \in [2]$) denotes the forward directed path that is a single edge, $\bP$ is a forward directed path of length $\ell$, and $r \geq 0$, then $\bQ$ is called an \emph{$r$-wave}. A $2$-wave is shown in Figure~\ref{2_wave}, \emph{1}.
\end{defn}

\begin{thm}\label{wave_theorem}
Let $\bQ$ be a wave. Then $\bQ$ has PSBPD, $\coCSP(\bQ)$ is definable in symmetric Datalog, and $\CSP(\bQ)$ is in $\ccL$.
\end{thm}

\begin{proof}
We prove the case when $\bQ$ is an $r$-wave for $r = 2$. For larger $r$-s, the proof generalizes in a straightforward manner. Let $\bP$ be a directed path of length $h$, $\bP_1, \bP_3, \bP_5$ be disjoint copies of $\bP$, and $\bP_2, \bP_4$ be copies of the reverse of $\bP$. Let $\bE_1$ and $\bE_2$ be forward edges. Assume the $2$-wave $\bQ$ is $\bE_1 \bP_1 \bP_2 \bP_3 \bP_4 \bP_5 \bE_2$ (Figure~\ref{two_wave_for_thm}).
\begin{figure}
\begin{center}
\includegraphics[scale=1]{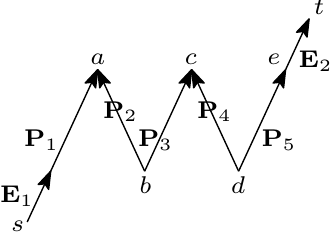}
\end{center}
\caption{2-wave in the proof of Theorem~\ref{wave_theorem}.}\label{two_wave_for_thm}
\end{figure}
We will provide a piecewise symmetric obstruction set $\cO_{ps}$ for $\bQ$, such that every element of $\cO_{ps}$ is an oriented path. To do this, first we observe that by \cite{Hell/Zhou:94:Homomorphisms}, $\bQ$ has path duality, i.e., we can assume that the set $\cO$ of all oriented paths that do not homomorphically map to $\bQ$ form an obstruction set for $\bQ$. To construct $\cO_{ps}$ from $\cO$, we will place certain elements of $\cO$ into $\cO_{ps}$ such that $\cO_{ps}$ is still an obstruction set for $\bQ$.

We begin with some simple observations. Any oriented path that has height at most $h+1$ maps to $\bQ$, so these oriented paths can be neither in $\cO$ nor in $\cO_{ps}$. Any oriented path that has height strictly larger than $h+2$ obviously does not map to $\bQ$, so all such paths are in $\cO$ and we also place these paths into $\cO_{ps}$. Assume that $\bP \in \cO$ has height exactly $h+2$. It is easy to see that if $\bP$ is not minimal, then it contains a minimal subpath that does not map to $\bQ$. Therefore, it is sufficient to place only those oriented paths from $\cO$ of height $h+2$ into $\cO_{ps}$ which are minimal.

Let $\bP \in \cO_{ps}$ of height $h+2$ (then $\bP$ is minimal). Intuitively, any attempt to homomorphically map the vertices of $\bP$ to $\bQ$ starting by first mapping the first vertex of $\bP$ to the first vertex of $\bQ$ and then progressively finding the image of the vertices of $\bP$ from left to right would get stuck at $a$ or $c$.

Formally, assume that the vertices of $\bP$ are $v_1,\dots,v_n$. Let $\bP_{[i]}$ denote the subpath of $\bP$ on the first $i$ vertices. Choose $i$ to be the largest index such that $\bP_{[i]} \overset{\varphi}{\longrightarrow} \bQ$ and $\varphi(v_1) = s$. Then $\varphi$ cannot be extended to $v_{i+1}$ for one of the following reasons. Clearly, $\varphi$ must map $v_i$ to a source or a sink other than $s$ or $t$, i.e., to $a$,$b$,$c$ or $d$. Furthermore, we can assume that $v_i$ is not mapped to $b$ or $d$. This is because if $v_i$ is mapped to $b$ or $d$, then $\lev(v_i)=1$, so the edge between $v_i$ and $v_{i+1}$ is from $v_i$ to $v_{i+1}$, and therefore $\varphi$ can be extended. So we can assume that $v_i$ is mapped to $a$ or $c$. Because we cannot extend $\varphi$, $v_{i+1}$ must be at level $\ell+2$, so it must be that $v_{i+1}$ is the last vertex $v_n$ of $\bP$. Because $\bP \not \rightarrow \bQ$, $\bP_{[n-1]}$ must be an oriented path such that any homomorphism $\varphi$ from $\bP_{[n-1]}$ to $\bQ$ such that $\varphi(v_1) = s$ maps $v_{n-1}$ to $a$ or $c$ but not to $e$.

We assume first that any homomorphism $\varphi$ from $\bP_{[n-1]}$ to $\bQ$ maps $v_{n-1}$ to $a$. We follow the vertices of $\bP_{[n-1]}$ from left to right. Let $w_a$ be the first vertex that is at level $h+1$. If there is a vertex to the right of $w_a$ at level $1$, then because $\bP_{[n-1]}$ will have to reach level $h+1$ again, we will be able to map $v_{n-1}$ to $c$, and that is not possible by assumption. So $\bP$ must have the following form (Form 1): $(w_1 \rightarrow w_2) \bX (w_3 \rightarrow w_4) \bY (w_5 \rightarrow w_6)$, where $\bX$ is any oriented path of height $h-1$ with first vertex at the bottom and last vertex at the top level of $\bX$, and $\bY$ is any oriented path of height $h-1$ with both its first and last vertices being in the top level of $\bY$. See Figure~\ref{2_wave_obs}, left.

For the second case, we assume that $\bP_{[n-1]}$ is such that $v_{n-1}$ can be mapped to $c$. Again, we follow the vertices of $\bP_{[n-1]}$ from left to right. Let $w_a$ be the first vertex that is at level $h+1$. We must have a vertex going back to level $1$ (otherwise we could not ``pass'' $b$ and could not map $v_{n-1}$ to $c$). Let $w_b$ be the first such vertex. We will have to go back to level $h+1$ again, so let $w_c$ be the first vertex at that level. Finally, we cannot go back to level $1$ again, since then the last vertex of $\bP_{[n-1]}$ can be mapped to $e$. We can ``go down'' to at most level $2$ of $\bP_{[n-1]}$. So $\bP$ must have the form (Form 2) $(w_1 \rightarrow w_2) \bX (w_3 \rightarrow w_4) \bY (w_5 \leftarrow w_6) \bZ (w_7 \rightarrow w_8) \bW (w_9 \rightarrow w_{10})$, where $\bX$ ($\bZ$) is any oriented path of height $h-1$ with first vertex at the bottom and last vertex at the top level of $\bX$ ($\bZ$), $\bY$ is any oriented path of height $h-1$ with first vertex at the top and last vertex at the bottom level of $\bY$, and $\bW$ is any oriented path of height $h-1$ with both its first and last vertices being in the top level of $\bW$. See Figure~\ref{2_wave_obs}, right.
\begin{figure}[h!]
\begin{center}
\includegraphics[scale=1]{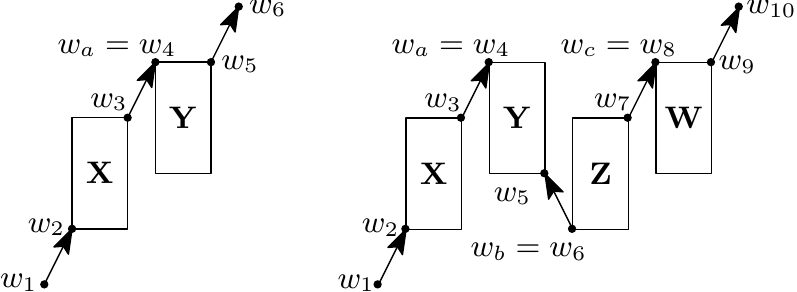}
 \caption{Obstructions of height $h+2$ for a $2$-wave.} \label{2_wave_obs}
\end{center}
\end{figure}

Because $\cO_{ps} \subseteq \cO$ and for any structure $\bS \in \cO$, there is a structure $\bS' \in \cO_{ps}$ such that $\bS' \rightarrow \bS$, $\cO_{ps}$ is an obstruction set for $\bQ$.

It remains to show that $\cO_{ps}$ is piecewise symmetric. Let $\bS$ be an oriented path of height more than $h+2$, and assume the vertex set of $\bS$ is $v_1,\dots,v_n$. We need to define a representation $\sS$, and a filter $\sF_{\bS}$ for $\bS$. The representation $(\bS_0,\bS_1,\dots,\bS_{n-2})$ is $(v_1,v_2), (v_2,v_3), \dots, (v_{n-1},v_n)$ (width $(1,2)$). The filter $\sF_{\bS}$ is the empty filter. Note that if we apply a zigzag operation to $\bS$, we get an oriented path of the same height as $\bS$, so $\cO_{ps}$ is closed under zigzagging of obstructions of height greater than $h+2$.

Let $\bS$ be an oriented path of height $h+2$ of Form 1, and assume the vertex set of $\bS$ is $v_1,\dots,v_n$. The representation $\sS = (\bS_0,\bS_1,\dots,\bS_{n-2})$ is constructed as in the previous paragraph. We specify $\sF_{\bS}$ to be the following $(3,6)$-filter. Assume that the edge $(w_3,w_4)$ is structure $\bS_i$. Then $\cF_{\bS} = \{[0,0],[i,i],[n-2,n-2]\}$. Using the definitions it is easy to see that if $\bP$ obeys the filter $\sF_{\bS}$, then $\xi(\sS,\bP)$ is also an oriented path of Form 1. Therefore $\cO_{ps}$ is closed under zigzagging of obstructions of Form 1. Obstructions of Form 2 can be handled similarly.
\end{proof}

We state the following generalization of waves.

\begin{defn}[Staircase]
A \emph{monotone wave} is an oriented path of the form $(\bar{\bP} \bP)^r \bar{\bP}$, where $\bP$ is a forward directed path and $r \geq 0$. We call the vertices of a monotone wave in the topmost level \emph{peaks}, and the vertices in the bottommost level \emph{troughs}.

If a minimal oriented path $\bQ$ can be expressed as $\bP_1 \bW_1 \bP_2 \bW_2 \dots \bP_{n-1} \bW_{n-1} \bP_n$, where $\bP_1,\dots,\bP_n$ are forward directed paths, $\bW_1,\dots,\bW_{n-1}$ are monotone waves, and for any $i \in [n-1]$, the troughs of $\bW_i$ are in a level strictly below the level of the troughs of $\bW_{i+1}$, and also, the peaks of $\bW_i$ are in a level strictly below the level of the peaks of $\bW_{i+1}$, then $\bQ$ is called a \emph{staircase}. An example is given in Figure~\ref{2_wave}, \emph{2}.
\end{defn}
\begin{thm}\label{stair_ps}
Let $\bQ$ be a staircase. Then $\bQ$ has PSBPD, $\coCSP(\bQ)$ is definable in symmetric Datalog, and $\CSP(\bQ)$ is in $\ccL$.
\end{thm}

\begin{proof}
Assume that the height $\bQ$ is $h$. As for waves, we use \cite{Hell/Zhou:94:Homomorphisms} to conclude that $\bQ$ has path duality. We will construct a piecewise symmetric obstruction set $\cO_{ps}$ for $\bQ$ by placing three classes of oriented paths into $\cO_{ps}$. First, $\cO_{ps}$ contains all oriented paths which have height strictly greater than $h$. These oriented paths obviously do not map to $\bQ$.

The next class of oriented paths we place into $\cO_{ps}$ are those which have height precisely $h$. Recall that $\bQ$ consists of waves patched together with directed paths in between. Let the wave subpaths of $\bQ$ be $\bW_1,\dots,\bW_n$, from left to right. For each $\bW_i$, we construct a class of oriented paths. Assume that $\bW_i$ has height $h_i$ and let $\cO_i$ be the set of minimal oriented paths of height $h_i$ which do not map to $\bW_i$. For each $\bR \in \cO_i$, we construct $\bC = \bB_1 \bR \bB_2$, where $\bB_1$ and $\bB_2$ are oriented paths (possibly empty) such that $\bC$ has height $h$, and the level of $\bR$ in $\bC$ matches the level of $\bW_i$. Observe that there cannot be a homomorphism from $\bC$ to $\bQ$. We place all such constructed $\bC$ into $\cO_{ps}$.

Let $\ell$ be the length of the longest directed subpath of $\bQ$. The third class of oriented paths are those that have height $h'$, where $\ell < h' < h$. For every such $h'$, we produce a set of obstructions. (\emph{Remark:} we set $\ell < h'$ because any oriented path of length $\ell$ or less maps to $\bQ$.)

Assume inductively (the base case is trivial) that we already have a piecewise symmetric obstruction set for every staircase of height strictly less than $h$. Consider every subpath $\bQ_1,\dots,\bQ_m$ of $\bQ$ of height $h'$. Notice that $\core(\bQ_i)$ is a staircase which is not a directed path. By the inductive hypothesis we have a piecewise symmetric obstruction set $\cU_i$ for $\bQ_i$. We keep only those oriented paths in $\cU_i$ which have height at most $h'$; observe that $\cU_i \neq \emptyset$. Construct $\bD = \bB_1 \bT_1 \cdots \bB_m \bT_m \bB_{m+1}$, where $(\bT_1, \dots, \bT_m) \in \cU_1 \times \dots \times \cU_m$ and the $\bB_j$ are arbitrary oriented paths such that the height of $\bD$ is $h'$. Place all these $\bD$-s into $\cO_{ps}$.

Notice that $\bD$ does not map to $\bQ$ for the following. Assume for contradiction that $\bD$ maps to a subpath $\bS$ of $\bQ$. Then $\bD$ also maps to the core of $\bS$ which is a staircase. But by construction $\bD$ contains a subpath that does not map to $\bS$.

We show that $\cO_{ps}$ is an obstruction set for $\bQ$. If an structure $\bZ \in \cO_{ps}$ homomorphically maps to an input structure $\bA$, then obviously, there cannot be a homomorphism from $\bA$ to $\bQ$. Assume for
contradiction that no structure in $\cO_{ps}$ maps to $\bA$ but $\bA$ does not map to $\bQ$. Then $\cO$ contains an oriented path $\bP$ that maps to $\bA$. So if we show the following claim then we are done.
\begin{inclaim}
For any oriented path $\bP$ that does not homomorphically map to $\bQ$, there is an oriented path $\bZ \in \cO_{ps}$ that homomorphically maps to $\bP$.
\end{inclaim}
\begin{proof}[Proof of Claim]
Assume that $\bP$ has height precisely $h$. We show that there exists $\bZ \in \cO_{ps}$ of height $h$ such that $\bZ \rightarrow \bP$. Assume for contradiction that none of the oriented paths of height $h$ in $\cO_{ps}$ map to $\bP$. As before, let $\bW_1,\dots,\bW_n$ be the wave segments of $\bQ$, from left to right, and assume without loss of generality that none of the $\bW_i$ is a directed path. Let the initial and final vertices of $\bW_i$ be $a_i$ and $b_i$ respectively, $i \in [n]$. For each $i \in [n]$, find the minimal oriented subpaths of $\bP$ whose initial vertices have the same level as $a_i$, and final vertices have the same level as $b_i$, or vice versa (note that because of the structure of $\bQ$, no such oriented path could contain another as a subpath, however, these oriented paths could overlap). For any such subpath $\bR$ of $\bP$ associated with $\bW_i$, map the lowest vertex of $\bR$ to $a_i$, and the highest vertex of $\bR$ to $b_i$. \emph{Remark 1: In fact there is no other choice.} The rest of the vertices of $\bR$ can be mapped to $\bQ$ as follows. If $\bR$ does not map to $\bW_i$ with first and last vertices matched then by definition, $\bP$ is in $\cO_{ps}$ and we have a contradiction. Therefore let the homomorphism for $\bR$ be $\varphi_\bR$. \emph{Remark 2: Also observe that $\varphi_\bR$ maps the inner vertices of $\bR$ to vertices of the staircase which are between $a_i$ and $b_i$.}

We show that the partial homomorphisms $\varphi_{\bR}$ map the same vertex of $\bP$ to the same vertex in $\bQ$, and furthermore we can also map those vertices of $\bP$ to an element of $\bQ$ that are not mapped anywhere by the $\varphi_{\bR}$. This way we obtain a homomorphism from $\bP$ to $\bQ$ and this would be a contradiction.

First, any vertex $v$ is assigned to a vertex of $\bQ$ by at most two homomorphisms which correspond to consecutive wave segments of $\bQ$. This is because in $\bQ$, $\bW_i$ and $\bW_j$ are disjoint unless $j=i+1$. Using Remarks 1 and 2, we can see that if a vertex $v$ of $\bP$ is in the domain of two ``non-consecutive'' homomorphisms, then because those homomorphisms could not agree on where to map $v$, it is not possible that $\bP \rightarrow \bQ$. This is a contradiction.

Let $\varphi_{\bR_1}$ and $\varphi_{\bR_2}$ (assume without loss of generality that $\bR_1$ and $\bR_2$ correspond to $\bW_1$ and $\bW_2$, respectively) be two partial homomorphisms such that their domains overlap. Then the markers $a_1,b_1,a_2,b_2$ appear in the order $a_1,a_2,b_1,b_2$ when traversing $\bP$ from left to right. The vertices that are in the domain of both homomorphisms are the ones from $a_2$ to $b_1$. By the choice of $a_1,b_1,a_2,b_2$, the segment of $\bP$ from $a_2$ to $b_1$ is a minimal oriented path. Checking the images of the vertices going back from $b_1$ to $a_2$ under the map $\varphi_{\bR_1}$, we see that these vertices are mapped to the rightmost directed path segment of $\bW_1$. Similarly, the image of these vertices under $\varphi_{\bR_1}$ is the leftmost directed path of the $\bW_2$. That is, the two homomorphisms coincide for the vertices from $a_2$ to $b_1$.

Furthermore, some vertices of $\bP$ are not in the domain of any partial homomorphisms. Consider the two minimal oriented paths $\bS$ and $\bS'$ on the two sides of such a maximal continuous sequence of vertices in $\bP$. There are two cases. First, assume that $\bS$ and $\bS'$ both correspond to the same $\bW_i$. Let the markers for $\bS$ be $a$ and $b$ an the markers for $\bS'$ be $a'$ and $b'$. Then following $\bP$ from left to right, the markers appear in the order $a,b,b',a'$. The images of the vertices from $b$ to $b'$ are not defined. (Observe that $b$ and $b'$ are mapped to the same vertex.) Consider the last directed path segment of $\bW_i$ together with the first directed path segment of $\bW_{i+1}$ (or just the last edges of $\bQ$ if $i=n$). Observe that the vertices from $b$ to $b'$ can be mapped to this directed path. The case when $\bS$ and $\bS'$ correspond to different waves of $\bQ$ is handled similarly.

Suppose lastly that $\bP$ has height $h'<h$. Because $\bP$ does not map to any of the subpaths of $\bQ$ of height $h'$, for each subpath $\bQ_1,\dots,\bQ_m$ of $\bQ$ of height $h'$, $\bP$ contains a subpath $\bS_i$ such that $\bS_i \not \rightarrow \bQ_i$, $i \in [m]$. If $\bS_i \not \rightarrow \bQ_i$ then $\bS_i \not \rightarrow core(\bQ_i)$. Recall that core($\bQ_i$) is a staircase and by definition, $\cU_i$ contains an oriented path $\bS_i'$ such that $\bS_i' \rightarrow \bS_i$. It is clear that we can choose oriented paths $\bB_1,\dots,\bB_{m+1}$ such that $\bB_1 \bS_1' \bB_2\dots\bB_m\bS_m' \bB_{m+1} \rightarrow \bP$.
\end{proof}
Finally, it is not hard to see from the construction how to associate filters with the elements of $\cO_{ps}$ to establish that $\cO_{ps}$ is piecewise symmetric.
\end{proof}

We also give a large class of oriented paths for which the $\CSP$ is $\ccNL$-complete. We need the following propositions to prove Theorem~\ref{NL_compl_thm}.

\begin{prop}\label{path_maps}
Let $\bP_1$ and $\bP_2$ be two minimal oriented paths of the same height $h$. Then there is a minimal oriented path $\bQ$ of height $h$ such that $\bQ \rightarrow \bP_1, \bP_2$.
\end{prop}
\begin{proof}
Not hard, see e.g.\ \cite{Haggkvist/et_al:88:Multiplicative}.
\end{proof}

\begin{prop}\label{core_ori_is_rigid}
A core oriented path has a single automorphism, i.e., it is rigid.
\end{prop}
\begin{proof}
Let $\bP$ be a core oriented path and $\bP$ be an isomorphic copy of $\bP'$. There are at most two isomorphisms from $\bP'$ to $\bP$ (because a vertex with indegree $0$ must be mapped to a vertex with indegree $0$, and similarly for a vertex with outdegree $0$). One possibility is to  map the first vertex of $\bP'$ to the first vertex of $\bP$ and the last vertex of $\bP'$ to the last vertex of $\bP$. For contradiction, assume that the second possibility happens, i.e., there is an isomorphism $\varphi$ that maps the first vertex of $\bP'$ to the last vertex of $\bP$ and the last vertex of $\bP'$ to the first vertex of $\bP$. Assume that both the first vertex $v$ and last vertex $w$ of $\bP'$ have indegree zero (the other case is similar). Then the $\lev(v) = \lev(w)$. This implies that the number of forward and backward edges in $\bP$ is the same, so $\bP$ has $2q$ edges. By the existence of $\varphi$, $\bP$ must have the form $\bQ \bar{\bQ}$, and such an oriented path is clearly not a core.
\end{proof}

\begin{thm}\label{NL_compl_thm}
Let $\bB$ be a core oriented path that contains a subpath $\bP_1 \bP_2 \bP_3$ of some height $h$ with the following properties: $\bP_1, \bP_2$ and $\bP_3$ are minimal oriented paths, they all have height $h$, and there is a minimal oriented path $\bQ$ of height $h$ such that $\bQ \rightarrow \bP_1$, $\bQ \rightarrow \bP_3$ but $\bQ \not \rightarrow \bP_2$. Then $\CSP(\bB)$ is $\ccNL$-complete.
\end{thm}
An example is given in Figure~\ref{2_wave}, \emph{3} and \emph{4}.

\begin{figure}[h!]
\begin{center}
\includegraphics[scale=1]{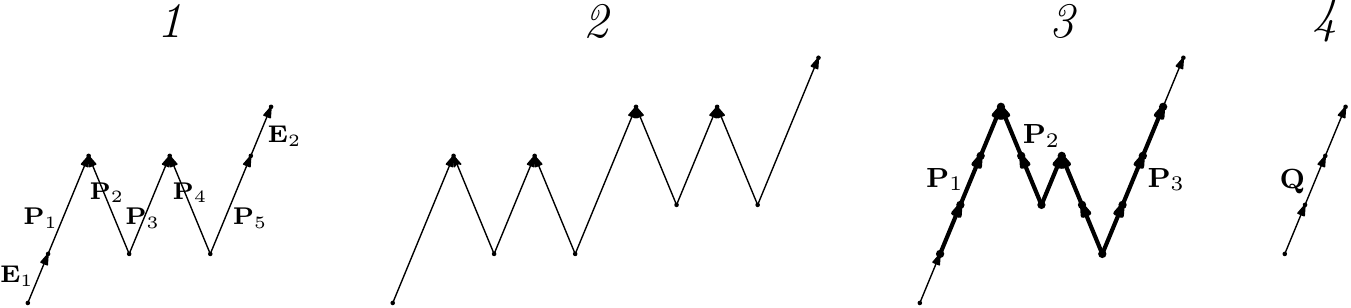}
 \caption{\emph{1:} A $2$-wave. \emph{2:} A staircase. \emph{3:} An example oriented path for which the $\CSP$ is $\ccNL$-complete. \emph{4:} The oriented path $\bQ$ in Theorem~\ref{NL_compl_thm} corresponding to the oriented path in \emph{3}.}\label{2_wave}
\end{center}
\vspace{-0.5cm}
\end{figure}

\begin{proof}[Proof of Theorem~\ref{NL_compl_thm}]
We show that the less-than-or-equal-to relation on two elements, $\bR_{\leq} = \{(0,0),(0,1),(1,1)\}$, and the relations $\{0\}$ and $\{1\}$ can be expressed from $\bP$ using a primitive positive (pp) formula (i.e., a first order formula with only existential quantification, conjunction and equality). It is easy to see and well known that $\CSP(\{\bR_{\leq},\{0\},\{1\}\})$ is equivalent to the $\ccNL$-complete directed $st$-\textsc{Conn} problem.

Since $\bP$ is a core, it is rigid by Proposition~\ref{core_ori_is_rigid}. Assume that the first vertex of $\bP_1$ is in a level lower than the level of the last vertex of $\bP_1$ (the other case can be handled similarly). See the illustration in Figure~\ref{NL_compl}. Assume that the first vertex of $\bP_1$ is $0$ and the first vertex of $\bP_3$ is $1$. We construct a structure $\bG$ with two special vertices $x$ and $y$ such that $\{(h(x),h(y)) \;|\; \bG \overset{h}{\longrightarrow} \bP\} = \bR_{\leq}$. It is well known and easy to show that then $\bR_{\leq}$ can also be expressed from $\bP$ using a pp-formula.
\begin{figure}[h!]
\begin{center}
\includegraphics[scale=1]{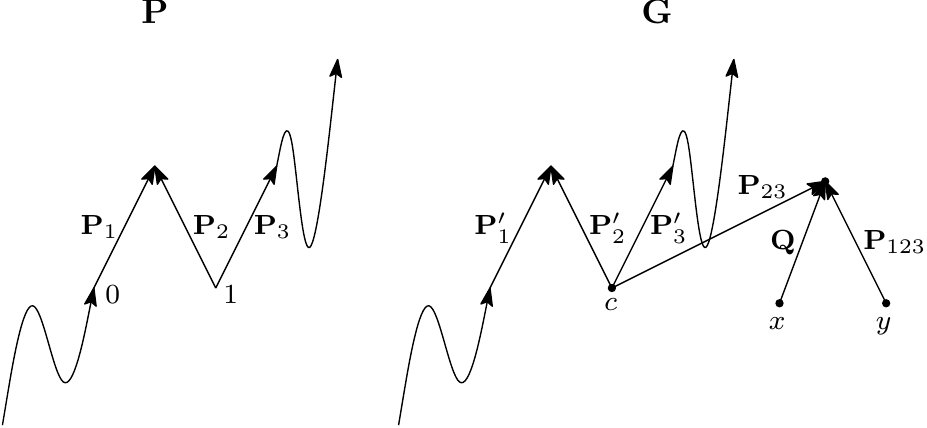}
 \caption{Construction of the gadget $\bG$.} \label{NL_compl}
\end{center}
\end{figure}
Let $\bP'$ be an isomorphic copy of $\bP$. We refer to copies of $\bP_1,\bP_2,\bP_3$ as $\bP_1',\bP_2',\bP_3'$, respectively. Using Proposition~\ref{path_maps}, we find a minimal oriented path $\bP_{23}$ of height $h$ that maps to both $\bP_2$ and $\bP_3$. Similarly, we find a minimal oriented path $\bP_{123}$ that maps to each of $\bP_1,\bP_2,\bP_3$. We rename the first vertex of $\bQ$ to $x$, and the first vertex of $\bP_{123}$ $y$ to $y$. To construct $\bG$, we identify the topmost vertices of the oriented paths $\bP_{23},\bQ$ and $\bP_{123}$. Then we identify the first vertex of $\bP_{23}$ with the vertex $c$ of $\bP'$ that is shared by $\bP_2'$ and $\bP_3'$. Observe that any homomorphism from $\bG$ to $\bP$, must map $c$ to $1$. It is straightforward to verify that $\{(h(x),h(y)) \;|\; \bG \overset{h}{\longrightarrow} \bP\} = \bR_{\leq}$.

Because $\bP$ is rigid, any relation of the form $\{v\}$ where $v \in P$ can be expressed by a pp-formula.
\end{proof}

\section{On CSPs in NL}\label{NL_section}

\subsection{Definitions}

Let $\tau$ be a vocabulary. A \emph{successor $\tau$-structure} $\bS$ is a relational structure with vocabulary $\tau \cup \{\one, \n, \suc\}$, where \one\ and \n\ are unary symbols and \suc\ is a binary symbol. Without loss of generality, the domain $S$ is defined as $\{1,\dots,n\}$, $\one^\bS = \{1\}$, $\n^\bS = \{n\}$, and $\suc^\bS$ contains all pairs $(i,i+1)$, $i \in [n-1]$. Because $\one^\bS, \n^\bS$ and $\suc^\bS$ depend only on $n$, they are called \emph{built-in} relations. When we say that a class of successor structures is homomorphism/isomorphism-closed, all structures under consideration are successor structures, and we understand that homomorphism/isomorphism closure, respectively, is required only for non-built-in relations.

\begin{defn}[Split Operation]
A \emph{split operation} produces a $\tau$-structure $\bA'$ from a $\tau$-structure $\bA$ as follows. For an element $a \in A$ let $T_a$ be defined as
\[
T_a = \{(\mathbf{t}, R, i) \;|\; \mathbf{t} = (t_1,\dots,t_r) \in R^\bA \text{ where } R \in \tau \text{, and } t_i = a\}.
\]
If $|T_a| \leq 1$ for every $a \in A$, then no  split operation can be applied. Otherwise we choose a strict nonempty subset $T$ of $T_a$, and for each triple $(\mathbf{t}, R, i) \in T$, we replace $\mathbf{t} = (t_1,\dots,t_r)$ in $R^\bA$ with $(t_1, \dots, t_{i-1}, a' ,t_{i+1}, \dots, t_r)$ to obtain $\bA'$ (and $A' = A \cup \{a'\}$).
\end{defn}
\begin{defn}[Split-Minimal, Critical]
Let $\cC$ be a class of structures over the same vocabulary. We say that a structure $\bA \in \cC$ is \emph{split-minimal in $\cC$} if for every possible nonempty sequence of split operations applied to $\bA$, the resulting structure is not in $\cC$. We say that a structure $\bA \in \cC$ is \emph{critical in $\cC$} if no proper substructure of $\bA$ is in $\cC$.

For a class of successor $\tau$-structures, criticality and split-minimality is meant only with respect to non-built-in relations.
\end{defn}

\begin{defn}[Read-Once Datalog]
Let $\cP$ be a (linear, symmetric) Datalog program that defines a class of structures $\cC$. If for every critical and split-minimal element of $\cC$ there is a $\cP$-derivation that is read-once, then we say that $\cP$ is \emph{read-once}.
\end{defn}

\begin{defn}[Read-Once mnBP1]
A \emph{monotone nondeterministic branching program} (mnBP) $H$ with variables $X = \{x_1,\dots,x_n\}$ computes a Boolean function $f_H : \{0,1\}^n \rightarrow \{0,1\}$. $H$ is a directed graph with distinguished nodes $s$ and $t$ and some arcs are labeled with variables from $X$ (not all arcs must be labeled). An assignment $\sigma$ to the variables in $X$ defines a subgraph $H_{\sigma}$ of $H$ as follows: an arc $a$ belongs to $H_{\sigma}$ if $\sigma(x) = 1$, where $x$ is the label of $a$, or if $a$ has no label. The function $f_H$ is defined as $f_H(\sigma) = 1$ if and only if there is a directed path in $H_{\sigma}$ from $s$ to $t$ (an \emph{accepting path}). The size of an mnBP is $|V_H|$.
\end{defn}

Let $\tau$ be a vocabulary and $n \geq 1$. We assume without loss of generality that any relational structure whose domain has size $n$ has domain $\{1,2,\dots,n\}$.  Let $(R_1,\mathbf{t}_1), (R_2,\mathbf{t}_2), \dots, (R_q,\mathbf{t}_q)$ be an enumeration of all pairs such that $R_i \in \tau$ and $\mathbf{t}_i \in \{1,2,\dots,n\}^{ar(R_i)}$. We associate a variable $x_i$ with $(R_i,\mathbf{t}_i)$, for each $i = 1,2,\dots,q$. Then if all labels of a branching program $H_n$ are among $x_1,x_2,\dots,x_q$, we say that $H_n$ is over the vocabulary $\tau$ for input size $n$. We say that a family of branching programs $\cF$ defines a class of $\tau$-structures $\cC$, if for each $n \geq 1$, $\cF$ contains precisely one branching program $H_n$ over $\tau$ for input size $n$ such that $f_{H_n}(x_1,x_2,\dots,x_q) = 1$ if and only if the tuple structure with domain $\{1,2,\dots,n\}$ and containing precisely those pairs $(R_i,\mathbf{t}_i)$ for which $x_i = 1$ is in $\cC$.

Let $\cF$ be a family of mnBP1s that contains precisely one branching program for each $n \geq 1$. We say that $\cF$ is a poly-size family if there is a polynomial $p$ such that for each $n \geq 1$, $|V(H_n)| \leq p(n)$. Such a family is denoted by mnBP1(poly). If for every $n$ and every structure of domain size $n$ in $\cC$, $H_n$ contains an accepting path $P$ such that any label on $P$ is associated with at most one arc of $P$, then we say that $\cF$ is \emph{read-once}. (This read-once condition can be made a bit weaker.)

\subsection{Examples}

We give some examples of problems definable by a 1-linDat(\suc) program or by an mnBP1(poly). The program in Section~\ref{Datalog}, Figure~\ref{example} without rule 3 is a read-once linear Datalog(\suc) program that defines the problem directed $st$-\textsc{Conn}. To see that this program $\cP_{st-\textsc{Conn}}$ is read-once, let $\bG$ be any input that is accepted (we do not even need $\bG$ to be critical and split-minimal). Then we find a directed path in $E^\bG$ connecting an element of $S^\bG$ to an element of $T^\bG$ without repeated edges. We build a $\cP_{st-\textsc{Conn}}$-derivation for this path in the obvious way.

For this section, by a clique we mean an ordinary undirected clique but each vertex may or may not have a self-loop. Let \textsc{EvenCliques} be the class of cliques of even size. The read-once linear Datalog(\suc) program $\cP_{\textsc{EC}}$ below defines \textsc{EvenCliques}. The goal predicate of $\cP_{\textsc{EC}}$ is $G_2$, and $E$ is the symbol for the edge relation of the input. The first part of $\cP_{\textsc{EC}}$ checks if the domain size $n$ of the input is even. The second part goes through all pairs $(x,y) \in [n]^2$, and at the same time, checks if $(x,y)$ is an edge in $E$. This is achieved by accessing the order on the domain. Program $\cP_{\textsc{EC}}$ goes through every pair of vertices precisely once, so every $\cP_{\textsc{EC}}$-derivation is read-once, and therefore $\cP_{\textsc{EC}}$ is read-once.

\begin{figure}[h!]
\begin{center}
\begin{eqnarray*}[rcl]
I(y) &\leftarrow& \one(x) \wedge \suc(x,y)\\
I(z) &\leftarrow& I(x) \wedge \suc(x,y) \wedge \suc(y,z) \\
G_1 &\leftarrow& I(x) \wedge \n(x) \\
J(x,y) &\leftarrow& G_1 \wedge \one(x) \wedge \one(y) \\
J(x,z) &\leftarrow& J(x,y) \wedge \suc(y,z) \wedge E(x,z) \wedge E(z,x) \\
J(z,w) &\leftarrow& J(x,y) \wedge \n(y) \wedge \suc(x,z) \wedge \suc(z,w) \wedge \\
&& E(z,w) \wedge E(w,z) \\
G_2 &\leftarrow& J(x,y) \wedge \suc(x,y) \wedge \n(y).\\
\end{eqnarray*}
\end{center}
\caption{The read-once linear Datalog(\suc) program $\cP_{\textsc{EC}}$ for \textsc{EvenCliques}.}
\end{figure}

In fact, we can easily test much more complicated arithmetic properties than the property of being even (e.g., being a power of $k$) with a 1-linDat(\suc) program. However, linear Datalog cannot define any set of cliques with a non-trivial domain size property in the following sense. Let $\bK$ be a clique of size $n$, and $\bK'$ be the clique obtained by identifying any two vertices of $\bK$. Then $\bK$ homomorphically maps to $\bK'$, and therefore if a linear Datalog program accepts $\bK$, then it also accepts $\bK'$. Therefore \textsc{EvenCliques} or, in fact, any set of cliques that contains a clique of size $n$ but no clique of size $n-1$ cannot be defined by a linear Datalog program. Since it is not difficult to convert a 1-linDat(\suc) program into an mnBP1(poly), the aforementioned problems can also be defined with an mnBP1(poly).

The additional power the successor relation gives to 1-linDat is at least twofold. For example, read-once linear Datalog(\suc) can do some arithmetic, as demonstrated above. In addition, let's define the \emph{density} of a graph to be the number of edges divided by the number of vertices. The density of an $n$-clique is $\binom{n}{2} / n = \theta(n)$. As demonstrated above, access to an order allows read-once linear Datalog(\suc) to accept only structures of linear density. On the other hand, any linear Datalog program $\cP$ accepts structures of arbitrary low density. For let $\bS$ be a structure accepted by $\cP$. Then adding sufficiently many new elements to the domain of $\bS$ yields a structure $\bS'$ whose density is arbitrarily close to $0$, and $\bS'$ is still accepted by $\cP$. One consequence of Corollary~\ref{main_coro1} is that if a read-once linear Datalog(\suc) defines $\coCSP(\bB)$, then both aforementioned additional abilities are of no use.

\subsection{Main Results}

We begin with stating the results for 1-linDat(\suc) and poly-size families of mnBP1s discussed in the Introduction.

\begin{thm}\label{main_thm1}
Let $\cC$ be a homomorphism-closed class of successor $\tau$-structures. If $\cC$ can be defined by a 1-linDat(\suc) program of width $(j,k)$, then every critical and split-minimal element of $\cC$ has a $(j,k+j)$-path decomposition.
\end{thm}

\begin{coro}\label{main_coro1}
If $\coCSP(\bB)$ can be defined by a 1-linDat(\suc) program of width $(j,k)$, then $\coCSP(\bB)$ can also be defined by a linear Datalog program of width $(j,k+j)$.
\end{coro}

\begin{thm}\label{main_thm2}
Let $\cC$ be a homomorphism-closed class of successor $\tau$-structures. If $\cC$ can be defined by a family of mnBP1s of size $O(n^j)$, then every critical and split-minimal element of $\cC$ has a $(j,r+j)$-path decomposition, where $r$ is the maximum arity of the symbols in $\tau$.
\end{thm}

\begin{coro}\label{main_coro2}
If $\coCSP(\bB)$ can be defined by a family of mnBP1s of size $O(n^j)$, then $\coCSP(\bB)$ can also be defined by a linear Datalog program of width $(j,r+j)$, where $r$ is the maximum arity of the relation symbols in the vocabulary of $\bB$.
\end{coro}

As discussed before, a wide class of $\CSP$s--$\CSP$s whose associated \emph{variety admits the unary, affine or semilattice types}--does not have bounded pathwidth duality \cite{Larose/Tesson:09:Universal}. It follows that all these $\CSP$s are not definable by any 1-linDat(\suc) program, or with any mnBP1 of poly-size. An example of such a $\CSP$ is the $\ccP$-complete $\CSP$ \textsc{Horn-3Sat}.

After some definitions, we give a high-level description of the proof of Theorem~\ref{main_thm1}. Any $\tau$-structure $\bM$ with domain size $n$ can be naturally converted into an isomorphic successor structure $\bM(\pi)$, where $\pi$ is a bijective function $\pi : M \rightarrow \{1,\dots,n\}$. We define the domain $M(\pi)$ as $\{1,\dots,n\}$ (note that this automatically defines $\one^{M_\pi}$, $\n^{M_\pi}$ and $\suc^{M_\pi}$) and for any $R \in \tau$, and $(t_1,\dots,t_{\ar(R)}) \in R^\bM$, we place the tuple $(\pi(t_1),\dots,\pi(t_{\ar(R)}))$ into $R^{\bM_\pi})$. When we want to emphasize that a structure under consideration is a successor $\tau$-structure, we use the subscript $s$, for example $\bM_s$. Given a successor $\tau$-structure $\bM_s$, $\bM$ denotes the structure $\bM_s$ but with the relations $\one^{\bM_s}$, $\n^{\bM_s}$ and $\suc^{\bM_s}$ removed.

We make the simple but important observation that we are interested only in isomorphism-closed classes. For example, $\coCSP(\bB)$ is obviously isomorphism-closed. We will crucially use the fact that if $\bM_s$ is accepted by a 1-linDat(\suc) program $\cP$, then $\cP$ must also accept $\bM(\pi)$ for any bijective function $\pi$. We are ready to describe the intuition behind the proof of Theorem~\ref{main_thm1}.

A 1-linDat(\suc) program that ensures that the class of successor-structures $\cC$ it defines is homomorphism-closed (and therefore isomorphism-closed) does not have enough ``memory''--due to its restricted width--to also ensure that some key structures in $\cC$ are ``well-connected''. If these key structures are not too connected, then we can define $\coCSP(\bB)$ in linear Datalog.

The more detailed proof plan is the following. Assume that $\coCSP(\bB)$, where the input is a successor structure, is defined by a linDat(\suc) program $\cP$ of width $(j,k)$. We choose a ``minimal'' structure $\bM$ in $\cC$ that is accepted, and assume for contradiction that $\bM$ does not have width $(j,k)$. Then roughly speaking, for all possible ``permutations of the domain elements of $\bM$'', $\bM$ must be accepted; therefore for each of these isomorphic structures, $\cP$ must be able to provide a derivation. Because this procedure will provide many enough derivations, we will be able to find some derivations which are of a desired form. The identification of these ``good'' derivations also crucially uses the generalized Erd\H os-Ko-Rado theorem. Once these derivations are detected, they can be combined to produce a derivation that ``encodes'' a structure of bounded pathwidth. The structures of bounded pathwidth produced this way can be used to define $\coCSP(\bB)$ in linear Datalog. We give the formal proofs.

We need the following additional definitions related to linear Datalog. In addition to extracting $\Ex(\sD)$ from $\sD$, we can also extract a decomposition of $\Ex(\sD)$ reminiscent of a path decomposition. For each $\ell \in [q]$, we define a tuple structure $\tilde{\bB}_\ell$ by adding $(R,\mathbf{t})$ to $\tilde{\bB}_\ell$ if $R(\mathbf{t})$ appears in $\rho_\ell$. In such a representation of $\Ex(\sD)$, we call $\tilde{\bB}_\ell$ the $\ell$-th \emph{bag}, and $(\tilde{\bB}_1,\dots,\tilde{\bB}_q)$ the \emph{tuple distribution} of $\Ex(\sD)$. It will be useful to remove empty bags from the list of bags $(\tilde{\bB}_1,\dots,\tilde{\bB}_q)$ to obtain the sequence $(\tilde{\bB}_{i_1},\dots,\tilde{\bB}_{i_t})$, where $i_\ell < i_{\ell'}$ if $\ell < \ell'$. For simpler notation, we renumber $(\tilde{\bB}_{i_1},\tilde{\bB}_{i_2}.\dots,\tilde{\bB}_{i_t})$ to $(\tilde{\bB}_1,\tilde{\bB}_2,\dots,\tilde{\bB}_t)$. We call the sequence $(\tilde{\bB}_1,\dots,\tilde{\bB}_t)$ the \emph{pruned tuple distribution} of $\sD$. The following is easy to prove.

\begin{prop}\label{hom_for_d}
Let $\bA'$ be a $\tau$-structure obtained from a $\tau$-structure $\bA$ by applying a sequence of split operations. Then $\bA' \rightarrow \bA$.
\end{prop}

We recall the following theorem tailored a bit to our needs.
\begin{thm}[Erd\H{o}s-Ko-Rado, general case; see, e.g., \cite{Frankl/Graham:89:Old}]\label{EKR}
Suppose that $\cF$ is a family of $s$-subsets of $\{1,\dots,n\}$, where $n \geq n_0(s,j+1)$. Suppose that for any two sets $S_1,S_2 \in \cF$, $|S_1 \cap S_2| \geq j+1$. Then $|\cF| \leq \binom{n-(j+1)}{s-(j+1)} = O(n^{s-(j+1)})$.
\end{thm}

\begin{proof}[Proof of Theorem~\ref{main_thm1}.]
Let the read-once linear Datalog(\suc) program that defines $\cC$ be $\cP$. Let $\bM$ be a structure in $\cC$ such that $\bM$ is critical and split-minimal, but assume for contradiction that $\bM$ has no $(j,k)$-path decomposition. Suppose that $M = \{m_1,\dots,m_s\}$. We choose a large enough $n$ divisible by $s$ (for convenience): how large $n$ should be will become clear later. We begin with constructing a class of successor structures from $\bM$. Let $\varphi : M \rightarrow \{1,\dots,n\}$ be a function that for all $i \in [s]$, maps $m_i$ to one of the numbers in
$\left\{(i-1) \cdot \frac{n}{s}+1,\dots, i \cdot \frac{n}{s} \right\}$. We call such a function an \emph{embedder}. Observe that there are $(\frac{n}{s})^s$ possible embedder functions. For each embedder $\varphi$, we define a successor structure $\bM_{\varphi}$ as follows. $\bM_{\varphi}$ is obtained from $\bM$ by renaming $m_i$ to $\varphi(m_i)$ for each $i \in [s]$, and adding all numbers inside $\{1,\dots,n\}$ but not in the range of $\varphi$ to the domain of the structure.

Obviously for any embedder $\varphi$, $\bM_\varphi$ contains an isomorphic copy of $\bM$, and therefore $\bM \rightarrow \bM_{\varphi}$. Since $\cC$ is closed under homomorphisms (and successor-invariant), it follows that for any embedder $\varphi$, $\bM_{\varphi}$ is accepted by $\cP$. Our goal now is to show that $\cP$ accepts a structure that can be obtained from $\bM$ by applying a nonempty sequence of split operations. This would contradict the split-minimality of $\bM$ with respect to $\cC$.

Let $\varphi_1,\dots,\varphi_t$ be an enumeration of all $t = (\frac{n}{s})^s$ embedders, and $\bM_{\varphi_1},\dots,\bM_{\varphi_t}$ the corresponding successor structures. Since $\cP$ is read-once, we can assume that for each $i \in [t]$, there is a read-once $\cP$-derivation for $\bM_{\varphi_i}$:
\[\mathscr{D}(\bM_{\varphi_i}) = (\rho_1^i, \lambda_1^i), \dots, (\rho_{q_i}^i,\lambda_{q_i}^i).\]

For each $\mathscr{D}(\bM_{\varphi_i})$ we denote its pruned tuple distribution as $(\tilde{\bB}_1^i,\dots,\tilde{\bB}_{w_i}^i)$. Let $\psi_i(\tilde{\bB}_1^i,\dots,\tilde{\bB}_{w_i}^i)$ denote $(\tilde{\bM}_1^i,\dots,\tilde{\bM}_{w_i}^i)$, where $\tilde{\bM}_\ell^i$ for each $\ell \in [w_i]$ is obtained as follows. For every $(R,\mathbf{t}) \in \tilde{\bB}_\ell^i$, place $(R, \varphi_i^{-1}(\mathbf{t}))$ into $\tilde{\bM}_\ell^i$. We call $\psi_i(\tilde{\bB}_1^i,\dots,\tilde{\bB}_{w_i}^i)$ the \emph{prototype} of $(\tilde{\bB}_1^i,\dots,\tilde{\bB}_{w_i}^i)$. We say that two pruned tuple distributions $(\tilde{\bB}_1^i,\dots,\tilde{\bB}_{w_i}^i)$ and $(\tilde{\bB}_1^{i'},\dots,\tilde{\bB}_{w_{i'}}^{i'})$ are \emph{similar} if they have the same prototypes, i.e., $\psi_i(\tilde{\bB}_1^i,\dots,\tilde{\bB}_{w_i}^i) = \psi_{i'}(\tilde{\bB}_1^{i'},\dots,\tilde{\bB}_{w_{i'}}^{i'})$.

Note that the codomain of $\psi_i$, for any $i$, is a sequence $S$ of bags such that a bag contains elements of $\tilde{\bM}$. Because by definition, every bag in $S$ is nonempty, and $\mathscr{D}(\bM_{\varphi_i})$ is read-once, we have that $|S| \leq |\tilde{\bM}|$. Therefore the number of possible bag sequences can be upper-bounded by a function of $s$; let this upper bound be $c_s$. It follows that there must be at least $t' = \frac{t}{c_s}$ embedders $\varphi_{i_1},\dots,\varphi_{i_{t'}}$ such that for any $\ell, \ell' \in \{i_1,i_2,\dots,i_{t'}\}$, $(\tilde{\bB}_1^\ell,\dots,\tilde{\bB}_{w_\ell}^\ell)$ and $(\tilde{\bB}_1^{\ell'},\dots,\tilde{\bB}_{w_{\ell'}}^{\ell'})$ are similar. Let the common prototype of all these similar pruned tuple distributions be $(\tilde{\bM}_1,\dots,\tilde{\bM}_w)$ $(i.e., \psi_{i_1}(\tilde{\bB}_1^{i_1},\dots,\tilde{\bB}_{w_{i_1}}^{i_1}))$. Because $\tilde{\bM}$ is critical, it follows that \mbox{$\tilde{\bM} = \tilde{\bM}_1 \cup \dots \cup \tilde{\bM}_w$}\footnote{Note that because $\tilde{\bM}$ is critical and $\cC$ is homomorphism closed, $\tilde{\bM}$ cannot contain isolated elements except when $\tilde{\bM}$ is a structure with a single element and no tuples. In this case the only critical and split-minimal element is $\tilde{\bM}$ and the empty set is a $(0,0)$-path decomposition for $\tilde{\bM}$.}.

To give a heads-up to the reader, our goal now is to construct a derivation $\mathscr{D}'$ using the derivations $\mathscr{D}(\bM_{\varphi_{i_1}}), \mathscr{D}(\bM_{\varphi_{i_2}}), \dots, \mathscr{D}(\bM_{\varphi_{i_{t'}}})$, such that $\Ex(\mathscr{D}')$ is isomorphic to a structure $\tilde{\bM}'$ that can be obtained from $\tilde{\bM}$ by a nonempty sequence of split operations. Because $\tilde{\bM}$ is split-minimal, this contradiction will complete the proof.

Define $X_g = \tilde{M}_1 \cup \dots \cup \tilde{M}_g$, and $Y_g = \tilde{M}_g \cup \dots \cup \tilde{M}_w$ for $g \in [w]$. If there is no $g \in [w-1]$ such that $|X_g \cap Y_{g+1}| > j$, then we construct a $(j,k+j)$-path decomposition $S_1,\dots,S_w$ for $\bM$ as follows. Define $S_1 = \tilde{M_1}$, $S_w = \tilde{M}_w$, and $S_\ell = \tilde{M}_\ell \cup (X_{\ell-1} \cap Y_{\ell+1})$, for $2 \leq \ell \leq w-1$. The first condition of Definition~\ref{path_decomp} is obviously satisfied. For the second condition, take $S_i$ and $S_{i'}$ and $i < \ell < i'$. If $a \in S_i \cap S_{i'}$ then $a \in \tilde{M}_{i''}$ and $a \in \tilde{M}_{i'''}$ for some $i'' \leq i$ and $i' \leq i'''$, so $a \in S_{\ell}$.  For the first part of the third condition observe that because $\cP$ has width $(j,k)$, $|\tilde{M}_\ell| \leq k$. Because we added at most $j$ new elements to $\tilde{M}_\ell$ to obtain $S_\ell$,  $|S_\ell| \leq k+j$ for any $\ell$. For the second part of the third condition, observe that $S_\ell \subseteq X_\ell$ and $S_{\ell+1} \subseteq Y_{\ell+1}$, so $|S_\ell \cap S_{\ell+1}| \leq j$ for any $\ell$.

For the other case, suppose that for some $g$, $|X_g \cap Y_{g+1}| > j$. Recall that for each $\ell \in \{i_1,i_2,\dots,i_{t'}\}$, $\tilde{\bM}_g$ was constructed from the bag $\tilde{\bB}^\ell_g$, and $\tilde{\bB}^\ell_g$ was constructed from a rule $\rho^\ell_{g_\ell}$ for some $g_\ell$, i.e., the $g_\ell$-th rule in the derivation $\mathscr{D}(\bM_{\varphi_{\ell}}) = (\rho^\ell_1,\lambda^\ell_1),\dots,(\rho^\ell_{q_\ell},\lambda^\ell_{q_\ell})$. Let $\iota$ be the number of IDBs of $\cP$ and $\kappa$ the maximum arity of any IDB of $\cP$. Recall that since $\cP$ has width $(j,k)$, any IDB contains at most $j$ variables.
Assume that the head IDB of $\rho^\ell_{g_\ell}$ is $I^\ell_{g_\ell}(\mathbf{x^{\ell}}_{g_\ell})$. Then there are at most $\iota j^\kappa n^j$ possibilities for the head IDB $I^\ell_{g_\ell}$ together with its variables instantiated to numbers in $[n]$. This means that there is an IDB $I$ and a tuple $\mathbf{t}$ such that for at least $t''=\frac{t'}{\iota j^\kappa n^j}$ values of $\ell \in \{i_1,i_2,\dots,i_{t'}\}$, it holds that
$I^\ell_{g_\ell} = I$, and $\lambda^\ell_{g_\ell}(\mathbf{x^{\ell}}_{g_\ell}) = \mathbf{t}$. Let these $t''$ values be $\{\ell_1,\dots,\ell_{t''}\}$.

We establish later that we can choose values $\ell_a, \ell_b \in \{\ell_1,\dots,\ell_{t''}\}$ such that the following inequality holds:
\begin{align}\label{small_intersection}
\left(\tilde{B}_1^{\ell_a} \cup \dots \cup \tilde{B}_w^{\ell_a} \right) \cap \left(\tilde{B}_1^{\ell_b} \cup \dots \cup \tilde{B}_w^{\ell_b} \right) \leq j.
\end{align}
Assuming that we have such $\ell_a$ and $\ell_b$, we define $\mathscr{D}'$ as:
\[(\rho^{\ell_a}_1,\lambda^{\ell_a}_1),\dots,(\rho^{\ell_a}_{g_{\ell_a}},\lambda^{\ell_a}_{g_{\ell_a}}),
(\rho^{\ell_b}_{g_{\ell_b}+1},\lambda^{\ell_b}_{g_{\ell_b}+1}),\dots,(\rho^{\ell_b}_{q_{\ell_b}},\lambda^{\ell_b}_{q_{\ell_b}}).\]
That is, we ``cut'' the derivations $\mathscr{D}(\bM_{\varphi_{\ell_a}})$ at the $g_{\ell_a}$-th rule, and cut the derivation $\mathscr{D}(\bM_{\varphi_{\ell_b}})$ at the $g_{\ell_b}$-th rule, and concatenate the first part of $\mathscr{D}(\bM_{\varphi_{\ell_a}})$ with the second part of $\mathscr{D}(\bM_{\varphi_{\ell_b}})$. $\mathscr{D}'$ is a valid derivation because at the point of concatenation, the head IDB of $\rho^{\ell_a}_{g_{\ell_a}}$ is the same as the IDB in the body of $\rho^{\ell_b}_{g_{\ell_b}+1}$, and the variables of this IDB are instantiated to the same values in both rules. Observe that the pruned tuple distribution of $\mathscr{D}'$ is
$(\tilde{\bB}_1^{\ell_a},\dots,\tilde{\bB}_g^{\ell_a},\tilde{\bB}_{g+1}^{\ell_b},\dots,\tilde{\bB}_w^{\ell_b})$. Set $\tilde{\mathbf{B}} = \tilde{\bB}_1^{\ell_a} \cup \dots \cup \tilde{\bB}_g^{\ell_a} \cup \tilde{\bB}_{g+1}^{\ell_b} \cup \dots \cup \tilde{\bB}_w^{\ell_b}$.

\begin{inclaim}
$\tilde{\mathbf{B}}$ is isomorphic to a structure that can be obtained from $\tilde{\bM}$ by a nonempty sequence of split operations.
\end{inclaim}
\begin{proof}[Proof of Claim]
Observe that the substructure $\tilde{\bM}_1 \cup \dots \cup \tilde{\bM}_g$ of $\tilde{\bM}$ is isomorphic to $\tilde{\bB}_1^{\ell_a} \cup \dots \cup \tilde{\bB}_g^{\ell_a}$ through $\varphi_{\ell_a}$. Similarly, $\tilde{\bM}_{g+1} \cup \dots \cup \tilde{\bM}_w$ is isomorphic to $\tilde{\bB}_{g+1}^{\ell_b} \cup \dots \cup \tilde{\bB}_w^{\ell_b}$ through $\varphi_{\ell_b}$. Our goal is to understand the difference between $\tilde{\bM}$ and $\tilde{\mathbf{B}}$.

Notice that because any embedder maps $m_i \in M$ into the interval $\left\{(i-1) \cdot \frac{n}{s}+1,\dots, i \cdot \frac{n}{s} \right\}$, and for any $i \neq i'$, $\left\{(i-1) \cdot \frac{n}{s}+1,\dots, i \cdot \frac{n}{s} \right\} \cap \left\{(i'-1) \cdot \frac{n}{s}+1,\dots, i' \cdot \frac{n}{s} \right\} = \emptyset$, if $i \neq i'$, then $\varphi_{\ell_a}(m_i) \neq \varphi_{\ell_b}(m_{i'})$. Therefore $\varphi_{\ell_a}$ and $\varphi_{\ell_b}$ can return the same value only if they both get the same input. The set $X_g \cap Y_{g+1}$ can be thought of as those elements of $\tilde{\bM}$ where $\tilde{\bM}_1 \cup \dots \cup \tilde{\bM}_g$ and $\tilde{\bM}_{g+1} \cup \dots \cup \tilde{\bM}_w$ are ``glued together'' to obtain $\tilde{\bM}$. Let $U = \tilde{B}_1^{\ell_a} \cup \dots \cup \tilde{B}_g^{\ell_a}$ and $V = \tilde{B}_{g+1}^{\ell_b} \cup \dots \cup \tilde{B}_w^{\ell_b}$. The set $U \cap V$ can be thought of as those elements of $\tilde{\bB}$ where $\tilde{\bB}_1^{\ell_a} \cup \dots \cup \tilde{\bB}_g^{\ell_a}$ and $\tilde{\bB}_{g+1}^{\ell_b} \cup \dots \cup \tilde{\bB}_w^{\ell_b}$ are ``glued together'' to obtain $\tilde{\bB}$.

If for all elements $m \in X_g \cap Y_{g+1}$, $\varphi_{\ell_a}(m) = \varphi_{\ell_b}(m)$, then $\tilde{\bB}$ would be isomorphic to $\tilde{\bM}$, i.e., $\tilde{\bB}_1^{\ell_a} \cup \dots \cup \tilde{\bB}_g^{\ell_a}$ would be glued to $\tilde{\bB}_{g+1}^{\ell_b} \cup \dots \cup \tilde{\bB}_w^{\ell_b}$ to obtain $\tilde{\bB}$ the same way as $\tilde{\bM}_1 \cup \dots \cup \tilde{\bM}_g$ is glued to $\tilde{\bM}_{g+1} \cup \dots \cup \tilde{\bM}_w$ to obtain $\tilde{\bM}$. But by Inequality~\ref{small_intersection}, $|X_g \cap Y_{g+1}| > |U \cap V|$. In other words, there are some elements $m \in X_g \cap Y_{g+1}$ which have one copy for $\varphi_{\ell_a}$, and another copy for $\varphi_{\ell_b}$ in $\tilde{\bB}$. Identifying $\varphi_{\ell_a}(m)$ and  $\varphi_{\ell_b}(m)$ for all such $m$ would convert $\tilde{\bB}$ to a structure isomorphic to $\tilde{\bM}$. Now it is easy to see that going backwards, splitting elements of $\tilde{\bM}$ would yield a structure isomorphic to $\tilde{\bB}$.
\end{proof}

It remains to show why we can choose $\ell_a$ and $\ell_b$ to satisfy Inequality~\ref{small_intersection}. Note that $t'' = \frac{(\frac{n}{s})^s}{c_s \iota j^\kappa n^j} \geq \Omega(n^{s-j})$. Also note that for any $\ell' \in \{\ell_1,\dots,\ell_{t''}\}$, $\tilde{B}_1^{\ell'} \cup \dots \cup \tilde{B}_w^{\ell'}$ is an $s$-subset of $[n]$. So by Theorem~\ref{EKR}, if for every pair $\ell_a,\ell_b \in \{\ell_1,\dots,\ell_{t''}\}$,
$\left(\tilde{B}_1^{\ell_a} \cup \dots \cup \tilde{B}_w^{\ell_a} \right) \cap \left(\tilde{B}_1^{\ell_b} \cup \dots \cup \tilde{B}_w^{\ell_b} \right) \geq  j + 1$,
then $t'' \leq O(n^{s-j-1)})$. But as observed $t'' \geq \Omega(n^{s-j})$, so for a large enough $n$ (as a function of $s$,$j$, $\iota$ and $\kappa$, so $n$ can be chosen in advance) Inequality~\ref{small_intersection} must hold for some $\ell_a,\ell_b \in \{\ell_1,\dots,\ell_{t''}\}$.
\end{proof}

\begin{proof}[Proof of Corollary~\ref{main_coro1}]
Let $\cO = \coCSP(\bB)$, i.e., the set of all those successor structures that do not homomorphically map to $\bB$. We construct an obstruction set $\cO'$ for $\bB$ such that every structure in $\cO'$ has pathwidth $(j,k+j)$. $\cO'$ is the set of all critical and split minimal structures of $\cO$. Theorem~\ref{main_thm1} tells us that every structure in $\cO'$ has a $(j,k+j)$-path decomposition.

To see that $\cO'$ is an obstruction set for $\bB$, take any structure $\bS \in \coCSP(\bB) = \cO$. Keep on applying split operations to $\bS$ and taking substructures of $\bS$ (again, these operations are with respect to non-built-in relations only), as long as the resulting structure is still in $\cO$. That is, if we apply any split operation to $\bS'$, or if we take any substructure of it, then the resulting structure is not in $\cO$ any more. Then $\bS' \in \cO'$ because $\bS'$ is critical and split minimal with respect to $\cO$. Using Proposition~\ref{hom_for_d}, we also see that $\bS' \rightarrow \bS$.

Because $\cO'$ is an obstruction set for $\bB$ such that every structure in $\cO'$ has width $(j,k+j)$, it follows from results of Dalmau in \cite{Dalmau:02:Constraint} that $\coCSP(\bB)$ is definable in linear $(j,k+j)$-Datalog.
\end{proof}

These proofs can be adapted for mnBP1s to obtain Theorem~\ref{main_thm2} and Corollary~\ref{main_coro2}.

\section*{Acknowledgement}
We thank Benoit Larose and Pascal Tesson for useful discussions and comments on an earlier draft. We also thank the anonymous referees for their helpful comments.

\bibliographystyle{abbrv}
\bibliography{mybib}

\end{document}